\documentclass[12pt]{oxarticle}

\pdfoutput=1

\usepackage{graphicx, float, array, xspace, amscd, amsmath, amsthm, amssymb, latexsym, relsize, bbm}
\usepackage[letterpaper, left=1.7in, right=0.3in, top=0.8in, bottom=1.2in]{geometry}
\usepackage[bbgreekl]{mathbbol}
\usepackage{amsfonts, tikz-cd}
\usepackage[T1]{fontenc}
\PassOptionsToPackage{linecolor=blue,backgroundcolor=blue!25,bordercolor=blue,textsize=scriptsize}{todonotes}
\usepackage{todonotes}
\usepackage{tikz}
\usetikzlibrary{backgrounds}
\usepackage[framemethod=tikz]{mdframed}
\AtBeginEnvironment{mdframed}{%
\tikzset{every picture/.style={}}%
}
\mdfsetup{roundcorner=.5ex}
{\begin{mdframed}[backgroundcolor=white]}%
{\end{mdframed}}

\DeclareSymbolFontAlphabet{\mathbb}{AMSb}
\DeclareSymbolFontAlphabet{\mathbbl}{bbold}

\usepackage[sort&compress, comma, square, numbers]{natbib}
\usepackage[all,cmtip]{xy}
\usepackage{tensor}
\usepackage{color}
\usepackage{slashed}
\definecolor{MyDarkBlue}{rgb}{0.15,0.25,0.45}
\usepackage[linktocpage=true]{hyperref}
\hypersetup{colorlinks=true, citecolor=blue, linkcolor=blue, urlcolor=blue, pdfauthor={}, pdftitle={},breaklinks=true}

\newcommand{\w}{{\,\wedge\,}}

\newcommand{\Lg}{{\mathfrak{g}}}

\newcommand{\vh}{{\vert h\vert}}
\newcommand{\vvh}{{\vert\vert h\vert\vert}}

\theoremstyle{definition}

\newtheorem{lemma}{Lemma}
\newtheorem{proposition}{Proposition}

\renewcommand{\d}{\delta}

\DeclareFontFamily{OT1}{pzc}{}
\DeclareFontShape{OT1}{pzc}{m}{it}{<-> s * [1.200] pzcmi7t}{}
\DeclareMathAlphabet{\mathpzc}{OT1}{pzc}{m}{it}

\DeclareFontFamily{U}{bbold}{}
\DeclareFontShape{U}{bbold}{m}{n}
 {  <-5.5> s*[1.05] bbold5
    <5.5-6.5> s*[1.05] bbold6
    <6.5-7.5> s*[1.05] bbold7
    <7.5-8.5> s*[1.05] bbold8
    <8.5-9.5> s*[1.05] bbold9
    <9.5-11.5> s*[1.05] bbold10
    <11.5-16> s*[1.05] bbold12
    <16-> s*[1.05] bbold17
 }{}

\font\eightrm=cmr8 at 8pt

\newcommand{\beq}{\begin{equation}}
\newcommand{\eeq}{\end{equation}}
\newcommand{\beqnn}{\begin{equation*}}
\newcommand{\eeqnn}{\end{equation*}}
\newcommand{\bea}{\begin{eqnarray}}
\newcommand{\eea}{\end{eqnarray}}
\newcommand{\bean}{\begin{eqnarray*}}
\newcommand{\eean}{\end{eqnarray*}}

\newcommand{\place}[3]{\vbox to0pt{\kern-\parskip\kern-7pt\kern-#2truein\hbox{\kern#1truein #3}\vss}\nointerlineskip}

\DeclareFontFamily{U}{wncy}{}
\DeclareFontShape{U}{wncy}{m}{n}{<->wncyr10}{}
\DeclareSymbolFont{mcy}{U}{wncy}{m}{n}
\DeclareMathSymbol{\sha}{\mathord}{mcy}{"58}

\newcommand{\del}{{\partial}}
\newcommand{\delb}{{\bar\partial}}
\newcommand{\bdel}{\boldsymbol{\partial}}
\newcommand{\bdelb}{\boldsymbol{\bar\partial}}

\newcommand{\Ah}{{\d A^{1,0}}}
\newcommand{\Aah}{{\d A^{0,1}}}

\newcommand{\End}{{\text{End}\,}}

\newcommand{\dd}{{\text{d}}}

\newcommand{\tr}{\text{tr}\hskip2pt}

\makeatletter
\g@addto@macro\bfseries{\boldmath}
\makeatother

\renewcommand{\baselinestretch}{1.1}
\numberwithin{equation}{section}
\proofmodefalse
\begin{document}
\pagestyle{empty}      
\ifproofmode\underline{\underline{\Large Working notes. Not for circulation.}}\else{}\fi

\begin{center}
\null\vskip0.2in
{\Huge  On a Deformed Holomorphic\\ Chern-Simons Theory\\[0.5in]}
{Eirik H{\o}gmoe Kjelsnes$^{a}$, Eirik Eik Svanes$^{a}$, \\
Vegard Undheim$^{a, b}$\\[0.5in]} 

{\it 
$^a$Department of Mathematics and Physics \\
Faculty of Science and Technology, University of Stavanger\\
N-4036, Stavanger, Norway\\[3ex]

$^b$Department of Computer Science and Computational Engineering \\
The Artic University of Norway\\
8505 Narvik, Norway\\[3ex]
}
%
\vspace{1cm}
{\bf Abstract\\[-8pt]}
\end{center}
We deform classical holomorphic Chern--Simons theory on a Calabi--Yau three-fold $X$ by deforming the complex structure by a deformation parameter $h\in{\cal H}^{0,1}(T^{1,0}X)$. The corresponding equations of motion admit new "instanton solutions" which which are invariant under re-scalings of $h$, and are perhaps more reminiscent of $G_2$-instantons for $G_2$ manifolds. We give examples of such instantons. In particular, when $h$ has non-vanishing Yukawa coupling ${\rm Yuk}(h,h,h)\neq 0$, it may be used to define a connection on $\End(T^{1,0}X)$ solving the instanton constraint. Interestingly, this connection gives rise to a hermitian (self-adjoint) connection for a real gauge theory on the real bundle $\End(TX)$ for only specific directions in deformation space, which may be classified using Morse theory. We quantize the deformed theory around these instanton backgrounds, and derive explicit expressions for the partition function in the limit where the complex structure deformation is large. We study anomalies, and the $h$-dependece of the partition function. In particular, coupling the theory to additional gravitational degrees of freedom, we find that the special directions in deformation space give rise to novel anomaly free theories on $\End(T^{1,0}X)$.

\vskip150pt

\newgeometry{left=1.5in, right=0.5in, top=0.75in, bottom=0.8in}
%
\newpage
{\baselineskip=10pt\tableofcontents}
\restoregeometry
\setcounter{page}{1}
\pagestyle{plain}
\renewcommand{\baselinestretch}{1.3}
\null\vskip-10pt

\section{Introduction}
Holomorphic Chern–Simons theory, or Donaldson--Thomas theory \cite{donaldson1998gauge, thomas1997a}, plays a central role at the interface between complex geometry, gauge theory, and string theory. Usually defined on Calabi–Yau threefolds, it governs the geometry of holomorphic vector bundles, where at the classical level, its equations of motion encode the condition that a gauge connection defines a holomorphic structure, while at the quantum level its partition function captures subtle holomorphic invariants of both the bundle and the underlying complex manifold. Because of this deep interplay between geometry and physics, understanding how holomorphic Chern–Simons theory responds to deformations of geometric structures is of significant interest.\footnote{This is particularly important in the context of heterotic string theory, where holomorphic Chern-Simons plays a vital role in understanding the moduli problem and corresponding moduli stabilisation \cite{Anderson:2010mh, Anderson:2011ty, Anderson:2014xha, delaOssa:2014cia, Garcia-Fernandez:2015hja, Ashmore:2018ybe, McOrist:2021dnd}.}

One of the most fundamental geometric inputs in the theory is the complex structure of the underlying Calabi–Yau manifold. Variations of complex structure affect not only the decomposition of differential forms into types but also the holomorphic top-form that appears explicitly in the action functional. Complex structure deformations are governed by elements of the Dolbeault cohomology group $H^{0,1}(T^{1,0}X)$, and their geometry is tightly constrained by special geometry relations and Yukawa couplings. From the physical perspective, these deformations correspond to marginal operators in the associated topological string theory and play a decisive role in moduli dependence and anomaly structure. It is therefore natural to ask how holomorphic Chern–Simons theory changes when the background complex structure is deformed, and what new geometric and quantum features may arise.

The dependence of the quantum partition function on the background complex structure and holomorphic anomaly equations have been studied extensively in the past. In the physics literature, such studies where initiated in \cite{Bershadsky:1993ta, Bershadsky:1993cx}, applying mathematical results of Bismut etal \cite{bismut1988analytic1, bismut1988analytic2, bismut1988analytic3}. In this work, we instead study a deformation of holomorphic Chern–Simons theory obtained by deforming the background complex structure explicitly in the classical action. Rather than working only with infinitesimal variations, we construct a fully deformed theory by expanding the holomorphic top-form in powers of a complex structure deformation parameter and inserting the resulting expression directly into the Chern–Simons functional. This produces a gauge theory whose action depends polynomially on the deformation parameter through a sequence of derived forms determined by special geometry. Although the resulting action can be interpreted as ordinary holomorphic Chern–Simons theory written in terms of a deformed complex structure, the explicit parametrization in terms of deformation tensors reveals new structural features that are not manifest in the standard formulation.

A consequence of the deformation is that the theory no longer depends solely on the $(0,1)$-part of the gauge connection. Instead, the deformed action couples both holomorphic and anti-holomorphic components, leading to modified equations of motion that mix curvature components of different Hodge types. We show that these equations admit a distinguished class of solutions that we refer to as scale-invariant instantons. These configurations are characterized by a holomorphic bundle condition together with an additional constraint involving contraction with the complex structure deformation tensor. Specifically, the equations of motion read
\begin{align}
    F^{0,2}&=0\\
    F^{1,1}\wedge\chi&=0\:,
\end{align}
where $\chi$ is the associated $(2,1)$-form of the complex structure deformation. Structurally, these instanton equations resemble higher-dimensional gauge-theoretic instanton conditions, and in particular share features with instanton equations arising in exceptional holonomy settings, particular for $G_2$ instantons. We present both abelian and non-abelian examples of such instantons, including a construction on endomorphism bundles of the tangent bundle.

An especially interesting phenomenon occurs when the deformation parameter has non-vanishing Yukawa coupling. In this case, the deformation tensor can be used to construct a natural Chern-type connection on the holomorphic tangent bundle whose induced connection on the endomorphism bundle satisfies the instanton constraint. However, compatibility with a real (or self-adjoint) gauge theory structure is not automatic: it holds only along special directions in complex structure deformation space. We show that these special directions can be characterized variationally in terms of a homogeneous functional built from Yukawa couplings and the moduli space metric. Using Morse-theoretic arguments, we derive bounds on the number of such directions and analyze their rigidity properties. This reveals a link between gauge-theoretic structures on vector bundles and the global geometry of complex structure moduli space.

We then proceed to study the quantum theory obtained by expanding around instanton backgrounds. Using the background field method, we show that after suitable field redefinitions the deformed theory can be rewritten in a form closely related to an anti-holomorphic Chern–Simons-type theory with a deformation-dependent kinetic operator. This reformulation makes it possible to carry out quantization using several complementary and standard approaches, including a formal one-loop analysis, BRST quantization, and the Batalin–Vilkovisky (BV) formalism. We compute the one-loop partition function and express it in terms of generalized Ray–Singer torsions associated with deformation-dependent differential operators. In the limit of large complex structure deformation, these expressions simplify and reduce to torsions associated with standard Dolbeault-type operators, multiplied by explicit powers of the "determinant" of the deformation tensor.

A central aspect of the quantum theory is its anomaly structure and geometric dependence. Because the deformation modifies both kinetic operators and measures, new potential gauge, gravitational and geometric anomalies arise. However, these simplify to more standard expressions in the limit where the complex structure deformation parameter is large. We analyze the anomalies systematically and identify conditions under which anomalies cancel. In particular, we find that the special directions in deformation space singled out by the hermiticity (or self-adjoint) condition give rise to anomaly-free theories on the trace-free endomorphism tangent bundle. This provides further evidence that these directions are geometrically and physically distinguished.

This paper is organized as follows. We begin in section \ref{sec:Theory} by defining the fully deformed holomorphic Chern–Simons theory and deriving its equations of motion, together with equivalent simplified formulations. We then analyze classes of instanton solutions  in section \ref{sec:examples}, and present explicit abelian and non-abelian examples, followed by a detailed study of the special deformation directions and their characterization via Morse theory on the projective moduli space. Next, in section \ref{sec:Quantum} we develop the quantum theory using the background field method and several quantization schemes, deriving explicit expressions for the partition function and its dependence on the deformation parameter. Finally, in section \ref{sec:anomalies} we investigate gauge, gravitational, and geometric anomalies and determine conditions for anomaly cancellation. Several appendices collect special geometry identities, properties of the deformation tensor, and curvature computations used throughout the paper.

\subsubsection*{Remark on hermitian vs. self-adjoint}
Before we begin, we remark on the difference between hermitian and self-adjoint connections, and how these terms will be used in this paper. In the literature, and the physics literature especially, the word hermitian is often used to describe a self-adjoint connection, even if it corresponds to a hermitian metric which is not positive definite, i.e. pseudo-hermitian. For clarity, let
\begin{equation}
    \dd_A=\delb+\del_A
\end{equation}
be a connection on a holomorphic bundle $V\rightarrow X$ over a complex manifold $X$. In this paper, we will only describe the connection as hermitian when 
\begin{equation}
    \del_A={\cal H}^{-1}\circ\del\circ{\cal H}\:,
\end{equation}
with respect to a {\it positive-definite} hermitian metric ${\cal H}$. If $\cal H$ is non-positive definite, we refer to the connection as pseudo-hermitian. We will also refer to the connection as self-adjoint in this case (but not neccesarily hermitian in the above sense).

\section{The theory}
\label{sec:Theory}
We begin by considering holomorphic Chern-Simons theory on a Calabi-Yau three-fold $X$ \cite{donaldson1998gauge, thomas1997a}. That is, we have a complex vector bundle $V\rightarrow X$ with a connection $A\in\Omega^1(\Lg)$, where $\Lg ={\rm Ad}(V)$, with the following action 
\begin{equation}
\label{eq:HolCS}
    S=\int_X\omega_{CS}(A)\wedge\Omega\:,
\end{equation}
where $\Omega$ is the Calabi-Yau holomorphic top-form $\Omega\in\Omega^{3,0}(X)$, and $\omega_{CS}(A)$ is the Chern-Simons three-form
\begin{equation}
    \omega_{CS}(A)={\rm tr}\left(A\dd A+\tfrac23A^3\right)\:.
\end{equation}
For ease of notation, we will often omit the wedge product between forms. 

We are interested in deforming the background complex structure of the theory \eqref{eq:HolCS}. That is, we write
\begin{equation}
    \Omega\rightarrow\Omega+\Delta\Omega=\Omega+\Omega(h)+\Omega(h,h)+\Omega(h,h,h)\:,
\end{equation}
where $h\in\Omega^{0,1}(T^{1,0}X)$ is the complex structure deformation, and
\begin{align}
    \Omega(h)&=\tfrac12h^a\wedge\Omega_{abc}\dd z^{bc}\\
    \Omega(h,h)&=\tfrac12h^ah^b\wedge\Omega_{abc}\dd z^c\\
    \Omega(h,h,h)&=\tfrac16h^ah^bh^c\,\Omega_{abc}\:.
\end{align}
We will take $h$ to be an on-shell complex structure deformation, and we will pick {\it harmonic gauge}. That is, $h$ is harmonic with respect to a choice of background metric given by a choice of Kähler form $\omega$. 

It follows from special geometry that $\Omega(h)$, $\Omega(h,h)$ and $\Omega(h,h,h)$ are all harmonic \cite{Candelas:1990pi, Strominger:1990pd}. In particular, $\Omega(h,h,h)$ is proportional to the anti-holomorphic top-form
\begin{equation}
    \Omega(h,h,h)=\vh\,\bar\Omega\:,
\end{equation}
where $\vh$ is a constant on $X$ which we will refer to as the determinant of $h$, point-wise viewed as a matrix. We give a formal proof of this in Appendix \ref{app:spacial}. Note that a generic complex structure deformation $h$ will have $\vh\neq0$, and we shall choose a complex structure deformation $h$ where this is the case. Note also that $\vh$ is proportional to the Yukawa coupling of $h$. Specifically, we have
\begin{equation}
\label{eq:vh-Yuk}
    {\rm Yuk}(h,h,h)=\int_X\Omega(h,h,h)\wedge\Omega=\vh\int_X\bar\Omega\wedge\Omega=i\vh e^{-K}\:,
\end{equation}
where the complex structure Kähler potential $K$ is given by
\begin{equation}
    e^{-K}=i\int_X\Omega\wedge\bar\Omega\:.
\end{equation}
This relates the determinant $\vh$ to special geometry.

We then have two theories of interest. The first theory is the infinitesimally deformed theory
\begin{equation}
\label{eq:InfTheory}
    S_1=\delta_hS=\int_X\omega_{CS}(A)\wedge\Omega(h)\:.
\end{equation}
The second theory is the fully deformed theory
\begin{equation}
    S_2=\int_X\omega_{CS}(A)\wedge(\Omega+\Delta\Omega)\:.
\end{equation}
We shall study the mathematical properties associated to these theories, in addition to their physical (quantum) properties. In this paper we will focus on the fully deformed theory $S_2$, as this theory, perhaps paradoxically, is easier to get a handle on.\footnote{This is perhaps not so paradoxical, as the theory $S_2$ is of course simply ordinary holomorphic Chern--Simons theory for the deformed complex structure $\tilde\Omega=\Omega+\Delta\Omega$, and so the reader might object that it is nothing new. Still, the given parametrization might provide further insight into the complex structure dependence of holomprphic Chern--Simons theory, which we explore in this paper.} We will leave the study of the infinitesimally deformed theory $S_1$ to a future publication.

\subsection{The fully deformed theory}
The fully deformed theory reads
\begin{equation}
    \label{eq:Full-def}
    S_2=\int_X\omega_{CS}(A)\wedge\left(\Omega+\Omega(h)+\Omega(h,h)+\Omega(h,h,h)\right)\:.
\end{equation}
In contrast to ordinary holomorphic Chern--Simons theory, which depends only on the $(0,1)$-part of the connection $A$, the action \eqref{eq:Full-def} depends both on the $(1,0)$- and $(0,1)$-parts of $A$. We will treat the $(1,0)$- and $(0,1)$-parts of the connection $A$ as {\it separate degrees of freedom}, leading to the equations of motion of the theory \eqref{eq:Full-def}
\begin{align}
F^{0,2}\wedge\Omega+F^{1,1}\wedge\Omega(h)+F^{2,0}\wedge\Omega(h,h)&=0
\label{eq:FullEOM1}
\\
F^{0,2}\wedge\Omega(h)+F^{1,1}\wedge\Omega(h,h)+F^{2,0}\wedge\Omega(h,h,h)&=0\:.
\label{eq:FullEOM2}
\end{align}
At first glance, these equations look somewhat complicated. We however have the following result. 
\begin{proposition}
    Equations \eqref{eq:FullEOM1} and \eqref{eq:FullEOM2} are equivalent to the following equations
    \begin{align}
    \label{eq:EOMfull1}
        F^{0,2}-F^{2,0}(h,h)&=0\\
    \label{eq:EOMfull2}
        2\,F^{0,2}+F^{1,1}(h)&=0\:,
    \end{align}
    where
    \begin{equation}
        F^{2,0}(h,h):=\tfrac12\, h^ah^b\,F^{2,0}_{ab}\:,\;\;\;
        F^{1,1}(h)=h\lrcorner F^{1,1} := h^a\wedge F_{a\bar b}\dd z^{\bar b}\:.
    \end{equation}
\end{proposition}
\begin{proof}
    We start by considering equation \eqref{eq:FullEOM1}. We apply $h$ to this equation in the sense that we contract holomorphic indices and wedge anti-holomorphic indices. We then get
    \begin{equation}
    \begin{aligned}
    &F^{0,2}(h)\wedge\Omega + F^{0,2}\wedge\Omega(h) + F^{1,1}(h)\wedge\Omega(h) \\ &+ F^{1,1}\wedge\Omega(h,h) + F^{2,0}(h)\wedge\Omega(h,h) + F^{2,0}\wedge\Omega(h,h,h) = 0 \ ,
    \end{aligned}
    \end{equation}
    and after subracting by \eqref{eq:FullEOM2} we are left with \begin{equation}
         F^{1,1}(h)\wedge\Omega(h) + F^{2,0}(h)\wedge\Omega(h,h) = 0 \ ,
    \end{equation}
    since $F^{0,2}(h)= h^{\bar a} F_{\bar a \bar b}\dd z^{\bar b}=0$ as $h$ has no anti-holomorphic vector components. We make the following index manipulation \begin{equation}
        \begin{aligned}
            0 &= (F_{d[e} \Omega_{abc]}) h^b h^a \dd z^{ec} \\ &= (F_{de} \Omega_{abc} - F_{da} \Omega_{ebc} - F_{db} \Omega_{aec} - F_{dc} \Omega_{abe}) h^b h^a \dd z^{ec} \\ &= (2F_{de} \Omega_{abc} - 2F_{db} \Omega_{aec}) h^b h^a \dd z^{ec} \ ,
        \end{aligned}
    \end{equation}
    from which it follows that \begin{equation}
    \begin{aligned}
        &F^{2,0}(h)\wedge\Omega(h,h) = -h^dF_{de} \w \tfrac{1}{2} h^b \w h^a \w\Omega_{abc} \dd z^{ec} \\ &= -h^d \w h^bF_{db}  \w \tfrac{1}{2} h^a \w \Omega_{abc} \dd z^{ec} = -F^{2,0}(h,h)\wedge\Omega(h) \ ,
    \end{aligned}
    \end{equation}
    and we see that $F^{1,1}(h)\wedge\Omega(h) = F^{2,0}(h,h)\wedge\Omega(h)$. By contracting $F^{1,1}(h) \w \Omega= F^{1,1} \w \Omega(h) =F^{2,0}(h) \w \Omega = F^{2,0} \w \Omega(h) = 0$ by $h$ we get the following relations \begin{align}
        F^{1,1}(h,h) \w \Omega + F^{1,1}(h) \w \Omega(h) &= 0 \\ F^{1,1}(h) \w \Omega(h) + F^{1,1} \w \Omega(h,h) &= 0 \label{eomrel2} \\ F^{2,0}(h,h) \w \Omega + F^{2,0}(h) \w \Omega(h) &= 0 \\ F^{2,0}(h) \w \Omega(h) + F^{2,0} \w \Omega(h,h) &= 0 \label{eomrel4} \ ,
    \end{align} 
    and contracting \eqref{eomrel2} and \eqref{eomrel4} by $h$ we get \begin{align}
        2F^{1,1}(h) \w \Omega(h,h) + F^{1,1} \w \Omega(h,h,h) &= 0 \\
        F^{2,0}(h,h) \w \Omega(h) + 2F^{2,0}(h) \w \Omega(h,h) + F^{2,0}(h) \w \Omega(h,h,h) &= 0 \ .
    \end{align}
    We then find that \eqref{eq:FullEOM2} becomes \begin{equation} \begin{aligned}
        &F^{0,2} \w \Omega(h) + F^{1,1} \w \Omega(h,h) + F^{2,0} \w \Omega(h,h,h) \\ &= F^{0,2} \w \Omega(h) - 2F^{1,1}(h) \w \Omega(h) - F^{2,0}(h,h) \w \Omega(h)  - 2F^{2,0}(h) \w \Omega(h,h) \\ &= F^{0,2} \w \Omega(h) - 3F^{1,1}(h) \w \Omega(h) - 2F^{2,0}(h) \w \Omega(h,h) \\ &= F^{0,2} \w \Omega(h) - 3F^{2,0}(h,h) \w \Omega(h) + 2F^{2,0}(h,h) \w \Omega(h) \\ &= F^{0,2} \w \Omega(h) - F^{2,0}(h,h) \w \Omega(h) = 0 \ ,
    \end{aligned}
\end{equation}
or, equivalently 
\begin{equation}
\label{eq:EOMp1}
        F^{0,2}  - F^{2,0}(h,h) = 0 \ .
\end{equation}
Applying this to \eqref{eq:FullEOM1} we see that it takes the form 
\begin{equation}
        \begin{aligned}
        &F^{0,2} \w \Omega + F^{1,1} \w \Omega(h) + F^{2,0} \w \Omega(h,h) \\ &=  F^{0,2} \w \Omega + F^{1,1} \w \Omega(h) + F^{2,0}(h,h) \w \Omega \\ &=2 F^{0,2} \w \Omega - F^{1,1} (h) \w \Omega = 0 \ ,
        \end{aligned}
\end{equation}
or, equivalently 
\begin{equation}
\label{eq:EOMp2}
        2F^{0,2}  - F^{1,1}(h) = 0 \ .
\end{equation}
The operations we have done to arrive at \eqref{eq:EOMp1} and \eqref{eq:EOMp2} are invertible, so this proves the proposition.
\end{proof}
A sufficient condition which solves \eqref{eq:EOMfull1}-\eqref{eq:EOMfull2} is to pick a holomorphic bundle, where
\begin{equation}
    \label{eq:Holomorphic}
    F^{0,2}=F^{2,0}=0\:,
\end{equation}
also satisfying the "instanton" condition
\begin{equation}
    \label{eq:Instanton}
    F^{1,1}\wedge\Omega(h)=0\:,
\end{equation}
or equivalenty
\begin{equation}
    F(h)=h\lrcorner F=0\:.
\end{equation}
This is the kind of background solution we are mostly interested in, also when we come to quantize the theory in section \ref{sec:Quantum}.

A few remarks are in order regarding the instanton solutions \eqref{eq:Holomorphic}-\eqref{eq:Instanton}. Firstly, we note that the solutions are in fact the only solutions of the equations of motion \eqref{eq:EOMfull1}-\eqref{eq:EOMfull2} if we insist that the solutions of \eqref{eq:FullEOM1}-\eqref{eq:FullEOM2} should be invariant under a rescaling of $h$ by some non-zero constant, which is also the type of solution that we will mostly be interested in. We refer to solutions of \eqref{eq:Holomorphic}-\eqref{eq:Instanton} as {\it scale invariant solutions}. 

Secondly, we also note that the scale invariant solutions are the only solutions if we insist the bundle be holomorphic. Furthermore, equation \eqref{eq:Instanton} interestingly also takes a similar form to instanton conditions appearing in exceptional geometry, and $G_2$ instantons in particular. Here the $(2,1)$-form $\Omega(h)$ plays the role of the non-degenerate, or integrable, $G_2$ three-form. Here, $\Omega(h)$ is integrable in the sense that the corresponding deformation tensor $h^a_{\bar b}$ is non-degenerate ($\vh\neq0$).

Thirdly, one might also wonder about existence and uniqueness of solutions of equations \eqref{eq:Holomorphic}-\eqref{eq:Instanton}, and a potential moduli space. We will return to study these equations and their moduli space further, when we study the infinitesimally deformed theory $S_1$ of equation \eqref{eq:InfTheory} in more detail in an upcoming paper. Indeed, the equations of motion of $S_1$ are
\begin{align}
\label{eq:instanton21}
    F^{0,2}\wedge\Omega(h)&=0\\
\label{eq:instanton22}
    F^{1,1}\wedge\Omega(h)&=0\:.
\end{align}
For $h$ non-degenerate, the first equation implies that the bundle is holomorphic, while the second is presicely the instanton constraint \eqref{eq:Instanton}. 

\subsection{Remarks on the moduli space of instantons}
As a prelude, let's briefly consider the infinitesimal moduli equations of the instanton equations \eqref{eq:instanton21}-\eqref{eq:instanton22}. Indeed, the corresponding moduli problem has interesting mathematical properties, especially if we treat the $(1,0)$ and $(0,1)$ parts of the deformation of $A$ as separate degrees of freedom. Briefly, if one defines $\beta=h^a \Ah_a$ and $\alpha = \delta A^{0,1}$, and construct the doublet $y=(\beta,\alpha)$ valued in $Q=\Lg\oplus\Lg$, one find that the infinitesimal deformations (with $h$ fixed) are in the kernel of the differential operator
\begin{equation}
    \bar{D} = \begin{pmatrix}
    \delb_A & h^a\nabla_a \\ 0 & \delb_A
\end{pmatrix} \ .
\end{equation}
It is straight forward to check that $\bar D$ squares to zero if and only if the instanton constraint \eqref{eq:instanton22} is satisfied. The infinitesimal moduli are then counted by the cohomology $H^{0,1}_{\bar D}(Q)$, where the associated deformation complex is given by
\begin{equation}
\label{eq:DefComplex}
    0\rightarrow\Omega^{0,0}(Q)\xrightarrow{\bar D}\Omega^{0,1}(Q)\xrightarrow{\bar D}\Omega^{0,2}(Q)\xrightarrow{\bar D}\Omega^{0,3}(Q)\rightarrow0\:.
\end{equation}
This deformation complex has several interesting features, which we remark on.

Firstly, the complex \eqref{eq:DefComplex} can be shown to have a vanishing index, or Euler characteristic \cite{Kjelsnes2025}. Hence the corresponding "virtual" or expected dimension of the moduli space also vanishes. Moduli problems where this is the case tend to be nicely behaved with respect to the definition of topological invariants, etc, such as for example ordinary Donaldson-Thomas theory on Calabi-Yau three-folds \cite{donaldson1998gauge, thomas1997a}.

Secondly, note that the operator $\bar D$ has an upper-triangular strucure, which in turn defines $Q$ as an extension sequence, though with a {\it non-tensorial} extension class $h^a\nabla_a$. Viewed in this way, the geometric structure $(Q,\bar D)$ does not define a holomorphic bundle in the usual sense, as would be the case if the extension class was tensorial, but rather a kind of holomorphic sheaf where the linear transition maps between local patches are promoted to differential operators. However, one may still relate the structure to algebraic geometry and Čech cohomology, and prove a Dolbeault theorem \cite{deLazari:2024zkg, Ibarra2025}.

Thirdly, the differential $\bar D$ is reminiscent of differentials which have lately appeared in the study of moduli problems in the context of heterotic string compactifications, where such non-tensorial extensions also appear \cite{McOrist:2021dnd, deLazari:2024zkg, Ibarra2025, Chisamanga:2024xbm, McOrist:2024ivz, Ashmore:2025fxr}. The simpler nature of the moduli problem at hand makes this a fruitful toy model for the study of such differential operators, and a playground for the more complicated analogs which appear in the context of heterotic moduli.

\section{Examples of instantons}
\label{sec:examples}
The reader may wonder if non-trivial solutions to the instanton equations \eqref{eq:Holomorphic}-\eqref{eq:Instanton} exist. In this section we answer this in the positive, though with slightly contrived solutions. We will present both an abelian example, and a non-abelain example on $\End(TX)$, focusing more on the properties of the non-abelian example in section \ref{sec:EndTtheory}. 

\subsection{An abelian example}
For an abelian example solving equations \eqref{eq:Holomorphic}-\eqref{eq:Instanton}, one can consider the ample holomorphic line bundle with a connection whose curvature is the Kähler form $\omega$. If $h$ is chosen {\it harmonic} with respect to the metric of $\omega$, then $\Omega(h)\in\Omega^{2,1}(X)$ is also harmonic. It follows that
\begin{equation}
    \omega\wedge\Omega(h)=0\:.
\end{equation}
Indeed, we can Lefschetz decompose $\Omega(h)$ as
\begin{equation}
    \Omega(h)=\Omega(h)_{\rm p}+a\wedge\omega\:,
\end{equation}
where $\Omega(h)_{\rm p}$ is the primitive part of $\Omega(h)$, that is
\begin{equation}
    \omega\lrcorner\Omega(h)_{\rm p}=0\:,
\end{equation}
and $a\in\Omega^{1,0}(X)$. Since the Laplacian commutes with the Lefschetz decomposition on a Kähler Calabi-Yau, it follows that if $\Omega(h)$ is harmonic, then $a$ is also harmonic, and thus must vanish as $h^{1,0}(X)=0$.

\subsection{A non-abelian example on $\End(T^{1,0}X)$ and $\End(TX)$}
\label{sec:EndTtheory}
For a non-trivial non-abelian example of a gauge theory satisfying equations \eqref{eq:Holomorphic}-\eqref{eq:Instanton}, one can consider the "Chern-type" connection $\del_{\tilde\nabla}$ on the tangent bundle given by $h$, defined in Appendix \ref{app:connection}, whose curvature satisfies \eqref{eq:Instanton}. As shown in the appendix, this gives a connection on $\End(T^{1,0}X)$, which solves the equations \eqref{eq:Holomorphic} and \eqref{eq:Instanton}. To connect this more with physics, it is desirable to view $\del_{\tilde\nabla}$ as the $(1,0)$-part of a real gauge theory on the real tangent bundle $\End(TX)$. The full connection is
\begin{equation}
\label{eq:FullConn}
    \dd_{\tilde\nabla}=\del_{\tilde\nabla}+\delb_{\tilde\nabla}\:,
\end{equation}
where $\delb_{\tilde\nabla}$ acts non-trivially on anti-holomorphic indices. We have to decide what this action is, and there are two natural choices. We will discuss both choices here.

\subsubsection*{Choice I}
In the first choice, which will be our main choice in this paper, we simply use $h$ again in the definition of the connection, where $\delb_{\tilde\nabla}$ acts as
\begin{equation}
    \delb_{\tilde\nabla}\beta^{\bar a}={h}_c^{\bar a}\delb\left({h}^c_{\bar b}\beta^{\bar b}\right)\:,
\end{equation}
where ${h}_c^{\bar a}$ is the inverse of ${h}^c_{\bar b}$. It is then straight forward to check that the full connection $\dd_{\tilde\nabla}$ remains an instanton (a solution to \eqref{eq:Instanton}). However, it is not the case that $\dd_{\tilde\nabla}$ is a self-adjoint connection on $\End(TX)$ for generic choices of $h$, in the sense that the $(0,1)$-part of $\dd_{\tilde\nabla}$ is the hermitian conjugate of the $(1,0)$-part, as required to have a real gauge theory on $\End(TX)$ as desired by physics. The constraint of a self-adjoint connection implies that
\begin{equation}
    {h}_c^{\bar a}\delb\left({h}^c_{\bar b}\right)={\bar h}_c^{\bar a}\delb\left({\bar h}^c_{\bar b}\right)\:.
\end{equation}
This equation can be re-arranged into
\begin{equation}
    \delb\left(h_{\bar b}^c\bar h^{\bar b}_{d}\right)=0\:.
\end{equation}
As $\End(T^{1,0}X)$ has no holomorphic sections other than a constant times the identity, we conclude that we must have
\begin{equation}
\label{eq:Inv=cc}
    h_a^{\bar b}=c\,\bar h_a^{\bar b}\:,
\end{equation}
where $c$ is $h$-dependent, but constant on $X$. Note that $c$ must be a real constant, as contracting with $h$ gives
\begin{equation}
    3=c\,\bar h_a^{\bar b}h_{\bar b}^a\:.
\end{equation}
If we further expand $h$ as
\begin{equation}
    h=Z^Ah_{A\,\bar b}^a\:,
\end{equation}
where $\{h_A\}$ form a basis of ${\cal H}^{0,1}(T^{1,0}X)$, where $\cal H$ denotes harmonic forms. We then get
\begin{equation}
    3=c\,Z^A\bar Z^{\bar B}h_{A\,a}^{\bar b}\bar h_{\bar B\,\bar b}^a=c\,Z^A\bar Z^{\bar B}G_{A\bar B}=c\,\vert Z\vert^2\:, 
\end{equation}
where $G_{A\bar B}$ is the metric on complex structure moduli space, and where we have used the special geometry relation $h_{A\,a}^{\bar b}\bar h_{\bar B\,\bar b}^a=G_{A\bar B}$, which we also prove in Appendix \ref{app:spacial}, proposition \ref{prop:GAB}. The special directions $h$ in the tangent space of the complex moduli space where the constraint \eqref{eq:Inv=cc} is satisfied will be studied in more detail below. We also also note that the connection symbols are trace-free (Proposition \ref{prop:TorsionFree}). The connection $\dd_{\tilde\nabla}$ therefore gives an $SU(3)$ gauge theory on $\End_0(TX)$, where the zero denotes the trace-free part.

\subsubsection*{Choice II}
The second choice of connection on anti-holomorphic indices, which we mention for completeness, is to use $\bar h$ instead of $h$, that is
\begin{equation}
    \delb_{\tilde\nabla}\beta^{\bar a}={\bar h}_c^{\bar a}\delb\left({\bar h}^c_{\bar b}\beta^{\bar b}\right)\:,
\end{equation}
where ${\bar h}^c_{\bar b}$ now is the inverse of the complex conjugate of $h$. This ensures that the connection $\dd_{\tilde\nabla}$ is self-adjoint for all choices of $h$ where $h_{\bar b}^a$ is invertible. The issue now is however that the full connection \eqref{eq:FullConn} does not in general satisfy the instanton constraint \eqref{eq:Instanton}. Denoting the curvature of the full connection $\dd_{\tilde\nabla}$ by $\cal\tilde R$, a straight-forward computation shows that the instanton constraint is equivalent to
\begin{equation}
     h^{\bar a}{\cal\tilde R}_{b\bar a}\dd z^{b}=0\:.
\end{equation}
For this choice, the non-trivial part of the instanton constraint becomes
\begin{equation}
    \dd z^{b}{\tilde R}_{b\bar a}{}^{\bar c}{}_{\bar d}h^{\bar a}=\dd z^{b}{\tilde R}_{b\bar d}{}^{\bar c}{}_{\bar a}h^{\bar a}=0\:,
\end{equation}
where we use that the connections involved are torsion-free, Proposition \ref{prop:TorsionFree}. Using that $\del h^{\bar a}=0$, as the inverse of a harmonic $h$ is also harmonic (see Proposition \ref{prop:inverse}), this equation becomes
\begin{equation}
    \del\delb_{\tilde\nabla}h^{\bar c}=0\:.
\end{equation}
This equation is clearly solved for the directions satisfying equation \eqref{eq:Inv=cc}. Indeed, both choices of connections agree for these special directions. One might wonder if these are the only solutions. A proof of this is beyond the scope of this paper, though we are tempted to conjecture that this is the case.

From hereon, we will restrict to {\it Choise I} when discussing the connection on ${\rm End}_0(TX)$. We shall also focus on the special directions satisfying \eqref{eq:Inv=cc}, which makes $\dd_{\tilde\nabla}$ a self-adjoint connection of a real $SU(3)$ gauge theory on ${\rm End}_0(TX)$ satisfying the instanton constraint \eqref{eq:Instanton}. The corresponding connection is then also metric, however with respect to a pseudo-hermitian metric which reads
\begin{equation}
\label{eq:HermMetr}
    {\cal H}^{\bar b}\,_{\bar a}\,^c\,_d=(\bar h h)^{\bar b}\,_{\bar a}\,^c\,_d=\bar h_d^{\bar b}h^c_{\bar a}=\tfrac{1}{c}\,h_d^{\bar b}h^c_{\bar a}\:.
\end{equation}
Note that the metric is not positive definite. Indeed, let $\alpha\in\End(T^{1,0}X)$, and compute
\begin{equation}
    \bar\alpha^{\bar a}{}_{\bar b}\,{\cal H}^{\bar b}\,_{\bar a}\,^c\,_d\,\alpha^d{}_c=\bar\alpha^{\bar a}{}_{\bar b}\bar h_d^{\bar b}h^c_{\bar a}\alpha^d{}_c=\bar\beta^{\bar a}{}_d\beta^d{}_{\bar a}\:,
\end{equation}
where $\beta^d{}_{\bar a}=h^c_{\bar a}\alpha^d{}_c$. Using the ordinary metric to lower and raise one index, this becomes
\begin{equation}
    \bar\alpha^{\bar a}{}_{\bar b}\,{\cal H}^{\bar b}\,_{\bar a}\,^c\,_d\,\alpha^d{}_c=\bar\beta_{ad}\beta^{da}\:.
\end{equation}
The right hand side is positive for symmetric $\beta$'s, and negative for anti-symmetric $\beta$'s. It forllows that the metric ${\cal H}$ is of indefinite signature $(6,3)$. As such, for mathematical rigor, this is not the metric we ought to use when defining positive definite inner-products and corresponding adjoints and Laplacians on the space of fields. We shall return to this issue below when we discuss the quantum theory and issues related to gauge fixing in Section \ref{sec:Quantum}. Before we do so, we will initiate an investigation into the properties of the special directions where \eqref{eq:Inv=cc} holds true.

\subsection{Properties of special directions}
We shall be interested throughout the paper in directions $Z^A$ in parameter space where equation \eqref{eq:Inv=cc} holds true, where the corresponding gauge theory on ${\rm End}_0(TX)$ is real. This is of course true if there is only one direction in complex moduli space, such as for the mirror quintic, but as we will see it turns out not to be true in general, but only for special directions $Z^A$ of the deformation parameter. 

Using the equation for the inverse of $h$ \eqref{eq:h-Inv}, derived in Appendix \ref{app:Geometry}, we find the equation we must satisfy for \eqref{eq:Inv=cc} to be true is then
\begin{equation}
\label{eq:Condition1}
    \bar h_{\bar B\,a}^{\bar b}G^{\bar BB}\del_B{\cal G}\propto\bar h_{\bar B\,a}^{\bar b} \bar Z^{\bar B}\:,
\end{equation}
where
\begin{equation}
    {\cal G}=\log(\kappa)\,,\;\;\kappa=\kappa_{ABC}Z^AZ^BZ^C\,,\;\;\;{\rm and}\;\;\;\kappa_{ABC}=\int_X\Omega(h_A,h_B,h_C)\wedge\Omega
\end{equation}
are the Yukawa couplings. We also define
\begin{align}
    \kappa_A&=\kappa_{ABC}Z^BZ^C\\
    \kappa_{AB}&=\kappa_{ABC}Z^C\:.
\end{align}
Equation \eqref{eq:Condition1} is satisfied if and only if 
\begin{equation}
    G^{\bar BB}\del_B{\cal G}\propto\bar Z^{\bar B}\:.
\end{equation}
To find the proportionality constant, we contract with $Z_{\bar B}$, and use that
\begin{equation}
    Z^A\del_A{\cal G}=3\:.
\end{equation}
We then find that $Z^A$ must satisfy the (non-holomorphic) equation
\begin{equation}
\label{eq:original}
    G^{\bar BB}\del_B{\cal G}=\frac{3}{\vert Z\vert^2}\bar Z^{\bar B}\:,
\end{equation}
where $\vert Z\vert^2=Z^AG_{A\bar B}\bar Z^{\bar B}$. This equation can be rewritten as
\begin{equation}
    \del_A\log\left(\frac{\kappa}{\vert Z\vert^6}\right)=0\:.
\end{equation}
We are therefore interested in critical points of the function $\frac{\kappa}{\vert Z\vert^6}$ in holomorpic directions, where we also insist that $\kappa\neq0$, which is equivalent to $\vh\neq0$. At such points, a critical point is equivalent to a point where
\begin{equation}
    \del_A\left(\frac{\kappa\bar\kappa}{\vert Z\vert^6}\right)=0\:.
\end{equation}
Rewriting the equation like this has the upshot that we are differentiating a real function, so holomorpic critical points are equivalent to ordinary critical points. Furthermore the equation is invariant under complex re-scalings $Z^A\rightarrow tZ^A$, for $t\in\mathbb{C}^*$, i.e. re-scaling is always a flat direction.\footnote{Note also that such a re-scaling leave the connection $\tilde\nabla_a$ of Appendix \ref{app:connection} invariant.}

It follows that we can think of the function 
\begin{equation}
    f(Z,\bar Z)=\frac{\kappa\bar\kappa}{\vert Z\vert^6}
\end{equation}
as a function on $\mathbb{P}^N$, where $N=h^{2,1}-1$. As the function $f$ is smooth without singularities on $\mathbb{P}^N$ which is compact, it must have a maximum. We therefore know that at least one solution exists, where also $\kappa\neq 0$. We would like to know if more such points exists, and if they are rigid, in the sense that a solution can only be infinitesimally deformed by rescaling. If the points are rigid in this sense, the points where the corresponding connections $\tilde\nabla$ on ${\rm End}_0(TX)$ are self-adjoint connections, and pseudo-hermitian given by the metric \eqref{eq:HermMetr}, will also be rigid.

\subsection{Bounds on special directions}
Using Morse theory, we can put some bounds on the minimum number of critical points of $f(Z,\bar Z)$. First note that we are not interested in the global minimum of $f$, given by the homogeneous degree three equation
\begin{equation}
    \kappa=\kappa_{ABC}Z^AZ^BZ^C=0\:.
\end{equation}
As a homogeneous polynomial of degree three, this equation describes a hypersurface of $\mathbb{P}^N$ of dimension $N-1$. Let us call this hypersurface $Y$.

Consider then a critical point of $f$ where $\kappa\neq0$. In a small open patch around the critical point, we pick coordinates on $\mathbb{P}^N$ so that $\kappa$ is constant. In these coordinates, a critical point of $f$ is then given as
\begin{equation}
    \delta Z^A G_{A\bar B}\bar Z^{\bar B}=0\,,\;\;\;\delta\kappa=3\,\delta Z^A\kappa_A=0\:.
\end{equation}
That is, $\delta Z^A\kappa_A=0$ for all directions $\delta Z^A$ which are orthogonal to $Z$. This implies that $\kappa^{\bar A}$ must be proportional to $\bar Z^{\bar A}$, which precisely leads to equation \eqref{eq:original}. For a deformation $\delta Z^A$ to be flat at second order, we find that we must also satisfy
\begin{equation}
    \delta Z^A G_{A\bar B}\,\delta\bar Z^{\bar B}=0\:,\;\;\;\delta Z^A\delta Z^B\kappa_{AB}=0\:.
\end{equation}
This shows that $\vert\delta Z\vert^2=0$, which can only happen if $\delta Z^A=0$. As this was the second order constraint, this implies that the Hessian of $f$ is non-degenerate at these critical points. I.e. $f$ is a Morse function outside of the vanishing locus $Y$, where $\kappa=0$. The corresponding critical loci are hence points on $\mathbb{P}^N$.

A lower bound on the number of critical points for such a Morse function, or Morse-Bott function, are computed by relative cohomology.\footnote{For an analogous situation, where Morse theory and Morse--Bott theory was used in the computation of chiral spectra local models of M-theory compactifications, see \cite{Pantev:2009de, Braun:2018vhk}.} That is, the cohomology of forms restricted to vanish on $Y$. The number of critical points $C^\gamma$ with index $\gamma$, where $\gamma$ are the number of negative eigenvalues of the Hessian, are then bounded by 
\begin{equation}
    C^\gamma\geq b^\gamma(\mathbb{P}^N,Y)\:,
\end{equation}
where $b^\gamma$ is the relative Betti number of order $\gamma$, i.e. the dimension of $H^\gamma(\mathbb{P}^N,Y)$. We will study the properties of the corresponding special directions, and connections to other features of Calabi--Yau moduli space, in more detail in future work. For now, we give an explicit low-dimensional example. 

\subsection{A low-dimensional example}
As an example, we consider the mirror of the bi-cubic CICY in the large complex structure limit. The configuration matrix is given by
\begin{align}
 X^*_{3}:   \left[
    \begin{array}{c|c} \mathbbm{P}^2 & 3 \\ 
    \mathbbm{P}^2 & 3 
    \end{array}
    \right]\:,
\end{align} 
with intersection numbers 
\begin{align}
    D_{111}= D_{222}=0\, ,\quad D_{112}= D_{122}=3\, .
\end{align} 
These then give the Yukawa couplings $\kappa_{ABC}$ of the mirror dual in the large complex structure limit. We therefore find that
\begin{equation}
    \kappa(Z)=3\,d_{112}(Z^{1})^2Z^2+3\,d_{122}Z^1(Z^2)^2=9\,Z^1Z^2(Z^2+Z^2)\:.
\end{equation}
The corresponding minimum locus $Y$ of $f(Z,\bar Z)$ as a function on $\mathbb{P}^1$ is then given by three points
\begin{equation}
    Z^1=0\:,\;\;\;Z^2=0\:,\;\;\;Z^1+Z^2=0\:,
\end{equation}
given as rays through the origin of $(\mathbb{C^*})^2$.

We can compute the Betti numbers $b^\gamma(\mathbb{P}^1,Y)$ by a long exact sequence in cohomology
\begin{align}
    0\rightarrow\; &H^0(\mathbb{P}^1,Y)\rightarrow H^0(\mathbb{P}^1)\rightarrow H^0(Y)\rightarrow\notag\\
    &H^1(\mathbb{P}^1,Y)\rightarrow H^1(\mathbb{P}^1)\rightarrow H^1(Y)\rightarrow\\
    &H^2(\mathbb{P}^1,Y)\rightarrow H^2(\mathbb{P}^1)\rightarrow H^2(Y)\rightarrow 0\notag\:.
\end{align}
Note that $H^0(\mathbb{P}^1,Y)=0$ as these correspond to constant functions which vanish at $Y$. We also have
\begin{equation}
    H^1(\mathbb{P}^1)=H^1(Y)=H^2(Y)=0\:,
\end{equation}
and 
\begin{equation}
H^0(\mathbb{P}^1)=H^2(\mathbb{P}^1)\cong\mathbb{R}\:,\;\;\;H^0(Y)\cong\mathbb{R}^3\:.
\end{equation}
The usual sequence chasing then gives the Betti numbers $b^1(\mathbb{P}^1,Y)=2$ and $b^2(\mathbb{P}^1,Y)=1$. It follows that the function $f$ will have at least three critical points outside of the global minimum $\kappa(Z)=0$, two saddle points and one maximum, corresponding to three distinct self-adjoint and pseudo-hermitian connections on ${\rm End}_0(TX)$.

The exact number of critical points can in general depend on the specifics of of the background complex structure. Let $\mu^A$ denote the moduli of the background complex structure. Writing the background complex structures as
\begin{equation}
    \mu^A=c^A+iw^A\:,
\end{equation}
the Kähler metric then reads
\begin{equation}
    G_{A\bar B}=\frac94\frac{\tilde\kappa_A\tilde\kappa_B}{\tilde\kappa^2}-\frac32\frac{\tilde\kappa_{AB}}{\tilde\kappa}\:,
\end{equation}
where 
\begin{equation}
\tilde\kappa=\kappa_{ABC}w^Aw^Bw^C\:,\;\;\;\tilde\kappa_A=\kappa_{ABC}w^Bw^C\:,\;\;\;\tilde\kappa_{AB}=\kappa_{ABC}w^C\:.
\end{equation}
As a demonstration, if we for example choose a background complex structure where
\begin{equation}
    (w^1,w^2)=(1,2)\:,
\end{equation}
we find three critical points, saturating the Morse theory bound. Choosing a normalization where $Z^1=1$, the critical solutions read
\begin{equation}
    Z^2=1\:,\;\;\;Z^2=-\tfrac15(7+2\sqrt{6})\:,\;\;\;Z^2=-\tfrac15(7-2\sqrt{6})\:.
\end{equation}
It is at this stage unclear if the number of critical directions can exceed the bound for specific choices of the background complex structure, also as we move away from the large complex structure limit. We hope to investigate this further in future works. 

\subsection{A covariant formulation}
Up until this point, we have fixed the complex structure background, determined by $\Omega$, and treat the $Z^A$ as directions in the tangent space of the complex structure moduli space ${\cal T}{\cal M}$, where $\cal M$ denotes the complex structure moduli space. Strictly speaking, the elements $h_A\in{\cal H}^{0,1}(T^{1,0}X)$ are covariant, transforming as components of a section of ${\cal T}^*{\cal M}$, and corresponding to elements
\begin{equation}
    {\cal D}_A\Omega=\chi_A=\Omega(h_A)\in{\cal H}^{2,1}(X)\:.
\end{equation}
That is,
\begin{equation}
    h_A^a=\frac{1}{2\vert\vert\Omega\vert\vert^2}\bar\Omega^{abc}\chi_{A\,bc\bar d}\dd z^{\bar d}\:.
\end{equation}
Note that even though $\Omega$ and $\chi_A$ also transform as sections of the holomorphic line bundle $\cal L$ over $\cal M$, this dependece drops out for $h_A$. So is we wish to retain covariance for $h=Z^Ah_A$ under change of coordinates of the background complex structure moduli space, we should promote 
\begin{equation}
    Z^A\rightarrow V^A\:,
\end{equation}
where $V^A$ are components of a section of $\cal{TM}$.

Note also that the Yukawa couplings 
\begin{equation}
    \kappa_{ABC}=\int_X{\cal D}_A{\cal D}_B{\cal D}_C\Omega\wedge\Omega\:,
\end{equation}
as elements of a section of $S^3({\cal T}^*{\cal M})$, also transform as sections of ${\cal L}\otimes{\cal L}={\cal L}^2$. Hence,
\begin{equation}
    \kappa(V)=\kappa_{ABC}V^AV^BV^C
\end{equation}
also transforms as a section of ${\cal L}^2$. The functional of $V$,
\begin{equation}
    f(V,\bar V)=\frac{\kappa(V)\bar\kappa(\bar V)}{\vert V\vert^6}
\end{equation}
therefore transforms as a section of ${\cal L}^2\otimes\bar{\cal L}^2$.\footnote{Note that ${\cal L}$ is a complex line bundle, where gauge transformations take values in $\mathbb{C}^*$. The bundle ${\cal L}^2\otimes\bar{\cal L}^2$ is therefore not trivial, as would be the case if $\cal L$ was a real $U(1)$ line bundle.} To get an invariant functional on $\cal M$, we multiply by $e^{2K}$ to get
\begin{equation}
    F(V,\bar V)=e^{2K}f(V,\bar V)=\frac{\vert\vert h(V)\vert\vert^2}{\vert V\vert^6}\:,
\end{equation}
by equation \eqref{eq:vh-Yuk}. Note in particular that $\vert h(V)\vert$ transforms as a section of ${\cal L}\otimes\bar{\cal L}^{-1}$, so $\vert\vert h(V)\vert\vert^2$ is invariant.

A vector field $V$ which extremizes the functional $F(V,\bar V)$, such that also $h(V)\neq0$, hence, at each point in the moduli space $\cal M$, corresponds to a self-adjoint and pseudo-hermitian connection of interest on ${\rm End}_0(TX)$. As noted above, at a fixed point in $\cal M$, $F(V,\bar V)$ will always have a maximum as a function on $\mathbb{P}^N$, so one might expect at least one such global section of $\cal TM$ to exist, modulo re-scaling.\footnote{Note however that even the position of the maximum in $\mathbb{P}^N$ need not be a smooth function over $\cal M$.} What happens to the other fixed points as we move around in the bacground complex structure moduli space, and if they can be extended to global sections, is beyond the scope of this paper, but will be analyzed further in upcoming work.

\section{The quantum theory and the background field method}
\label{sec:Quantum}
Let us now consider the quantum theory of \eqref{eq:Full-def}. We will do so using the background field method by deforming the theory \eqref{eq:Full-def} around a background solution to the equations of motion. We again restrict to considering instanton backgrounds, i.e. solutions the scale invariant solutions of \eqref{eq:Holomorphic} and \eqref{eq:Instanton}. For the purpose of physics, and the study of the quantum theory and anomalies, we shall from hereon assume that the background gauge theory is defined using a real gauge group with a self-adjoint connection. In particular, for the guage theory on ${\rm End}_0(T^{1,0}X)$ we shall restrict to the special directions of $h$ studied above, where \eqref{eq:Inv=cc} is satisfied and the corresponding connections are self-adjoint and pseudo-hermitian with respect to the (non-positive definite) metric \eqref{eq:HermMetr}.

Let us define the new fields
\begin{align}
\Aah &= \alpha\:,\\
h^a \Ah_a &= \beta\:.
\label{eq:rot1}
\end{align}A somewhat lengthy but straight forward computation then shows that the kinetic term of the theory may be written as

\begin{equation}
    S_{\rm kin}=\int_X( \alpha - \beta )(\delb_A - \d_A )(\alpha -\beta  )\wedge\Omega\:,
\end{equation}

Here the operator $\delta_A$ is defined as
\begin{equation}
    \delta_A=h^a\nabla_a\:,
\end{equation}
with $\nabla_a$ the gauge connection on gauge indices. Note that $\delta_A$ is nilpotent, that is for $\alpha\in\Omega^{0,p}(\Lg)$
\begin{equation}
    \delta_A^2\alpha=h^a\nabla_a\left(h^b\nabla_b\alpha\right)=h^a\tilde\nabla_a\left(h^b\tilde\nabla_b\alpha\right)=h^ah^b\tilde\nabla_a\tilde\nabla_b\alpha=0\:,
\end{equation}
where $\tilde\nabla$ is the "Chern-type" connection on holomorphic tangent bundle indecies defined in Appendix \ref{app:connection}. We have used that $h$ is "metric" with respect to $\tilde\nabla$, and the indices $a$ and $b$ are anti-symmetrised in the last equality. Useful for later, we also note that $\delta_A$ anti-commutes with $\delb_A$,
\begin{equation}
    \left(\delta_A\delb_A+\delb_A\delta_A\right)\alpha=h^a\nabla_a\delb_A\alpha-h^a\delb_A\nabla_a\alpha=h^a[F_a,\alpha]=0\:,
\end{equation}
where $F_a=F_{a\bar b}\dd z^{\bar b}$, where $F=\delb_A\del_A+\del_A\delb_A$ is the curvature. The last equality follows from $h^aF_a=0$, which is equivalent to \eqref{eq:Instanton}.

Finally, a somewhat lengthy but straight-forward computation shows that the cubic term of the theory takes a rather simple form
\begin{equation}
    S_{\rm cubic}=\tfrac23 \int_X \tr\left((\d A)^3\right) \wedge \left( \Omega+\Omega(h)+\Omega(h,h)+\Omega(h,h,h) \right) = \tfrac23 \int_X {\rm tr}\left((\alpha - \beta)^3\right) \w \Omega\:.
\end{equation}

The simple form of these terms prompts us to find a field redefinition for which $\tilde\beta=\alpha-\beta$. If we then redefine the fields as
\begin{align}
\label{eq:rot2-1}
    \tilde\beta&=-\beta+\alpha\\
\label{eq:rot2-2}
    \tilde\d_A&= \delb_A -\d_A\:,
\end{align}
the cubic term simplifies to
\begin{equation}
    S_{\rm cubic}=\tfrac23 \int_X {\rm tr}\left((\alpha - \beta)^3\right) \w \Omega=\tfrac23 \int_X \tr\left(\tilde\beta^3\right) \w \Omega = \tfrac13 \int_X \tr\left(\tilde\beta\wedge[\tilde \beta, \tilde \beta]\right)\w\Omega\:,
\end{equation}
and so the full theory becomes
\begin{equation}
    S_2=\int_X\tr\left(\tilde \beta \tilde \d_A \tilde \beta +\tfrac13\tilde\beta\wedge[\tilde \beta, \tilde \beta]\right)\wedge\Omega\:.
\end{equation}

If we now rotate $\tilde\beta$ as
\begin{equation}
\label{eq:rot3}
    \tilde\beta=-h^a\gamma_a\:,
\end{equation}
we see that the theory can be written in the following simple form as
\begin{equation}
    \label{eq:Full-def-redef}  S_2=\vh\int_X\tr\left(\gamma(\del_A-h^{\bar a}\nabla_{\bar a})\gamma+\tfrac13\gamma\wedge[\gamma,\gamma]\right)\wedge\bar\Omega\:,
\end{equation}
where $h^{\bar a}=\dd z^{b}{h_{b}}^{\bar a}$, and ${h_{b}}^{\bar a}$ is the inverse of ${h_{\bar b}}^{a}$.

We recognize the first part of the kinetic term, together with the qubic interaction term, as simply anti-holomorphic Chern-Simons theory around a background, multiplied by a factor of the determinant $\vh$ of $h$. It should however be noted that the connections $\delb_A$ and $\del_A$ appearing in each term are also $h$-dependent, as they derive from solutions of the background equation \eqref{eq:Instanton}. We will come back to this issue below. Finally, it is perhaps not too surprising that the theory takes such a simple form, depending only on the linear combination given by \eqref{eq:rot2-1}. Indeed, our starting theory is merely ordinary holomorphic Chern--Simons, on a Calabi--Yau with complex structure given by $\tilde\Omega=\Omega+\Delta\Omega$,
\begin{equation}
    S_2=\int_X\omega_{CS}(A)\wedge\tilde\Omega\:,
\end{equation}
rewritten around a background solving \eqref{eq:Holomorphic}-\eqref{eq:Instanton}. Of course, ordianry Holomorphic Chern--Simons only depends on the $(0,1)$-part of the connection, in the appropriate complex structure. However, the explicit parametrization in terms of the deformation tensor might reveal new structural features that are not manifest in the standard formulation, and that we wish to study here. 

The theories \eqref{eq:Full-def} and \eqref{eq:Full-def-redef} are equivalent at the classical level. However, when we quantize we need to be careful with potential Jacobians associated to the field redefinitions we have done. The natural quantum fields are given by $\delta A^{1,0}$ and $\delta A^{0,1}$. Note that the first rotation, \eqref{eq:rot1}, is the inverse of the third rotation \eqref{eq:rot3} modulo an irrelevant sign. The respective Jacobians therefore cancel. The Jacobian of second rotation, \eqref{eq:rot2-1}, at most gives an unimportant numerical factor. We are therefore free to treat \eqref{eq:Full-def-redef} as our quantum theory.

In the rest of the section we will consider different ways of quantizing the theory \eqref{eq:Full-def-redef}, to obtain the one--loop partition function. The methods we use are equivalent, modulo subtleties related to each quantization method which we explore, such as for example how to include zero-modes. Most of this is fairly standard, and the reader may be tempted to skip to section \ref{sec:anomalies}. However, we do make some important observations in the following subsections. In particular Proposition \ref{prop:twist} with its following observations, which allow for a streamlined $h$-dependence of the partition function for the ${\rm End}_0(T^{1,0}X)$-theory for backgrounds with self-adjoint connections corresponding to the special directions of complex structure deformations. 

\subsection{The quantum theory: Formal quantization}
\label{sec:QuantFormal}
We will warm up by considering the formal computation of the one-loop partition function of \eqref{eq:Full-def-redef}. The kinetic operator of the theory reads
\begin{equation}
    {\cal D}_A=\vh(\del_A-h^{\bar a}\nabla_{\bar a}):=\vh D_A\:,
\end{equation}
defining the operator $D_A$. Note that $D_A$ is nilpotent, as
\begin{align}
    D_A^2&=\del_A^2-\del_A\circ h^{\bar a}\nabla_{\bar a}-h^{\bar a}\nabla_{\bar a}\circ\del_A+h^{\bar a}\nabla_{\bar a}\circ h^{\bar b}\nabla_{\bar b}\notag\\
    &=\del_A^2-h^{\bar a}F_{\bar a}-\left(\del h^{\bar a}-\tfrac12[h,h]^{\bar a}\right)\nabla_{\bar a}\:.
\end{align}
The first term vanishes as $\del_A$ squares to zero. It is straigh forward to show that the vanishing of the second term is equivalent to the instanton equation \eqref{eq:Instanton}. The third term vanishes due to the Maurer--Cartan equation for $h^{\bar a}$. Indeed, It follows from Proposition \ref{prop:inverse} in Appendix \ref{app:Geometry} that $h^{\bar a}$ is harmonic, so
\begin{equation}
    \del h^{\bar a}=0\:,\;\;\;\tfrac12[h,h]^{\bar a}=0\:,
\end{equation}
where the last equality follows from a short computation. 

Ignoring harmonic modes in the path integral, a formal computation of the absolute value of the one-loop partition function then gives
\begin{equation}
    \vert Z_2^{\rm{\small formal}}\vert=\frac{{\rm det}'(\Delta_{\cal D_A}^{0,0})^{\tfrac34}}{{\rm det}'(\Delta_{\cal D_A}^{1,0})^{\tfrac14}}\:,
\end{equation}
where ${\rm det}'$ denote Zeta-regularised determinants, and $\Delta_{\cal D_A}^{p,0}$ is the Laplacian of ${\cal D}_A$ on  Lie-algebra valued $(p,0)$-forms, defined as
\begin{equation}
    \Delta_{\cal D_A}={\cal D}_A{\cal D_A}^\dagger+{\cal D_A}^\dagger{\cal D}_A=\vvh^2\left(D_AD_A^\dagger+D_A^\dagger D_A\right)\:,
\end{equation}
where the adjoints are defined by picking appropriate metrics on the manifold and gauge bundle. We will come back to the issue of gauge fixing below.  

Note then that given an elliptic operator $\cal O$, under a re-scaling by some constant $a$, the regularised determinant scales as
\begin{equation}
    {\rm det}'(a{\cal O})=a^{\zeta_{\cal O}(0)}{\rm det}'({\cal O})\:,
\end{equation}
where $\zeta_{\cal O}(s)$ is the associated Zeta-function of $\cal O$, to be defined more explicitly in examples below. It follows that the partition function may be written as
\begin{equation}
    \vert Z_2^{\rm{\small formal}}\vert=\vert\vh\vert^{-\tfrac12\left(\zeta_{\Delta^{1,0}_{D_A}}(0)-3\zeta_{\Delta^{0,0}_{D_A}}(0)\right)}\frac{{\rm det}'(\Delta_{D_A}^{0,0})^{\tfrac34}}{{\rm det}'(\Delta_{D_A}^{1,0})^{\tfrac14}}\:,
\end{equation}
$\vvh^2=\vh\bar\vh$. We define the holomorpic Ray--Singer torsion of $D_A$ as
\begin{equation}
    {\cal I}_{\rm RS}(D_A)=\left(\prod_{p=0}^3{\rm det}'(\Delta^{p,0}_{D_A})^{(-1)^{p+1}p}\right)^{\tfrac12}\:,
\end{equation}
and use the Serre-duality relations of Lemma \ref{lem:Serre} below, we see that
\begin{equation}
    \vert Z^{\rm{\small formal}}_2\vert=\vert\vh\vert^{-\tfrac12\left(\zeta_{\Delta^{1,0}_{D_A}}(0)-3\zeta_{\Delta^{0,0}_{D_A}}(0)\right)}{\cal I}_{\rm RS}(D_A)^{\tfrac12}\:.
\end{equation}
In the limit where the deformation parameter is sent to infinity, or $\vh\rightarrow\infty$, or the "large complex structure limit", the expression simplifies even further. Expressions involving the inverse of ${h_{\bar a}}^b$ will then tend to zero, and we simply get
\begin{equation}
\label{eq:FormalPartLargeh}
    \lim_{\vh\rightarrow\infty}\vert Z^{\rm{\small formal}}_2\vert=\vert\vh\vert^{-\tfrac12\left(\zeta_{\Delta^{1,0}_{\del_A}}(0)-3\zeta_{\Delta^{0,0}_{\del_A}}(0)\right)}{\cal I}_{\rm RS}(\del_A)^{\tfrac12}\:,
\end{equation}
where ${\cal I}_{\rm RS}(\del_A)$ is the the Ray--Singer torsion of $\del_A$. 

In the case when the connection $\del_A$ is hermitian with an associated positive definite hermitian metric, using the corresponding hermitian metric in the definition of the adjoint, the Ray--Singer torsion of $\del_A$ becomes simply the standard holomorphic Ray--Singer torsion. We can then also say more about the zeta-invariants appearing in the exponential. Using Lemma \ref{lem:Zeta} below, and ignoring the Hodge numbers as we are ignoring harmonic forms, for large $\vh$ we then get
\begin{equation}
    \vert Z^{\rm{\small formal}}_2\vert=\vert\vh\vert^{-\tfrac{{\rm dim}(\Lg)}{480}\chi(X)}{\cal I}_{\rm RS}(\del_A)^{\tfrac12}\:.
\end{equation}
This expression however does not reveal the full $h$-dependence of the partition function. In particular, the connection $\del_A$ has an implicit dependence on $h$ through the instanton constraint \eqref{eq:Instanton}, and therefore so does the Ray--Singer torsion. We will investigate this dependence further in the next section, where we will perform a more rigorous BRST and BV-quantization of the theory, and in section \ref{sec:anomalies}, when we come to discuss anomalies.  

\subsection{The quantum theory: BRST quantization}
We will now perform a more rigorous quantization of the action $S_2$, using the BRST formalism. The BRST quantization procedure is standard, and mimics that of three-dimensional Chern-Simons theory \cite{witten1989quantum, Axelrod:1991vq}, though with a slightly strange kinetic operator. We we introduce fermionic ghost fields $c$ and $\bar c$, and a gauge-fixing bosonic field $b$. The ghost is a zero-form, while $\bar c$ and $b$ are $(3,0)$-forms, all valued in $\Lg$. We'll use the notation
\begin{align}
{\cal D}_A&=\vh(\del_A-h^{\bar a}\nabla_{\bar a})=\vh D_A\\
\{\;,\,\}&=\vh[\;,\,]\:,
\end{align}
where $[\;,\,]$ is the standard Lie bracket super-commutator. That is for $\gamma_1\in\Omega^{n_1,0}(\Lg)$ and $\gamma_2\in\Omega^{n_2,0}(\Lg)$ of fermionic degree $m_1\in\{0,1\}$ and $m_2\in\{0,1\}$ respectively, the bracket reads
\begin{equation}
    [\gamma_1,\gamma_2]=\gamma_1\wedge\gamma_2-(-1)^{n_1n_2+m_1m_2}\gamma_2\wedge\gamma_1\:.
\end{equation}
The BRST algrbra then reads
\begin{align}
    Q\gamma&=-{\cal D}_A c-\{\gamma,c\}\:,\\
    Q c&=\tfrac12\{c,c\}\:,\\
    Q \bar c&= b\:,\\
    Q b&=0\:,
\end{align}
Recall that $Q$ is fermionic. It is easily checked that $Q^2=0$.

A common gauge fixing condition condition is Lorentz gauge, ${\cal D}_A^\dagger\gamma=0$. The adjoint is defined by picking appropriate metrics on the manifold and gauge bundle. The natural choice of metric on the manifold is arguably the Kähler metric of the background. Choosing a metric on the bundle is a bit more subtle, especially if the connection is only pseudo-hermitian, corresponding to a non-positive definite metric on the bundle. We shall return to this below when discussing gauge fixing of the explicit ${\rm End}_0(T^{1,0}X)$ theory. For stable bundles, one might expect the unique hermitian Yang-Mills metric to be the natural choice, but this is perhaps a bit too naive given the problem at hand.

The gauge fixing is implemented by the gauge fixed action
\begin{align}
    S_{2,\,{\rm gf}}&=S_2+ Q\int_X\tr\left(\bar c\wedge {\cal D}_A^\dagger\gamma\right)\wedge\bar\Omega\notag\\
    &=S_2+\int_X\tr\left(b\wedge {\cal D}_A^\dagger\gamma+\bar c\wedge {\cal D}_A^\dagger {\cal D}_A c+\bar c\wedge {\cal D}_A^\dagger\{\gamma,c\}\right)\wedge\bar\Omega
    \notag\\
    &=S_2+\int_X\tr\left(\bar\vh\,b\wedge D_A^\dagger\gamma+\vvh^2\,\bar c\wedge D_A^\dagger D_A c+\vvh^2\,\bar c\wedge D_A^\dagger[\gamma,c]\right)\wedge\bar\Omega\:,
\end{align}
where $\vvh^2=\vh\bar\vh$.

The partition function is then given as
\begin{equation}
    Z_2=\int{\cal D}\gamma{\cal D}c{\cal D}\bar c{\cal D}b\: e^{iS_{2,\,{\rm gf}}(\gamma,c,\bar c,b)}\:.
\end{equation}
It is then natural to absorb $\vh$ by performing a field redefinition 
\begin{equation}
\label{eq:fieldredef1}
    (\gamma, c, \bar c, b)\rightarrow(\tfrac{1}{\sqrt{\vvh}}\gamma, \tfrac{1}{\vvh}c, \tfrac{1}{\vvh}\bar c, \tfrac{1}{\sqrt{\vvh}}b)\:.
\end{equation}
This gives a somewhat more familiar looking action
\begin{equation}
\label{eq:S2b-redef}
    S_{2,\,{\rm gf}}=\int_X\tr\left(\gamma\, e^{i\theta}D_A\gamma+b\, e^{-i\theta}D_A^\dagger\gamma+\bar c\,D_A^\dagger D_A c+\tfrac{e^{i\theta}}{3\sqrt{\vvh}}\gamma\,[\gamma,\gamma]+\tfrac{1}{\sqrt{\vvh}}\bar c\,D_A^\dagger[\gamma,c]\right)\wedge\bar\Omega\:,
\end{equation}
save from the somewhat unconventional kinetic operator, the appearance of the phase of $\vh$ given by $e^{i\theta}$, and the cubic couplings have been re-scaled. 

However, we also get a Jacobian factor from the infinite-dimensional measure. What this factor is depends upon how we choose to regulate the path-integral. A somewhat naive parameterization of the measure in position space leads under the field redefinition to
\begin{align}
    \int{\cal D}\gamma{\cal D}c{\cal D}\bar c{\cal D}b&=\int\small\prod_x\dd\gamma_x\dd c_x\dd\bar c_x\dd b_x\notag\\
    \rightarrow\vvh^{\tfrac12{\rm dim}(\Lg)(-3+2+2-1)\sum_x1}\int\small\prod_x\dd\gamma_x\dd c_x\dd\bar c_x\dd b_x&=\vvh^{-{\rm dim}(\Lg)\cdot0\cdot\sum_x1}\int{\cal D}\gamma{\cal D}c{\cal D}\bar c{\cal D}b\notag\\
    &=\int{\cal D}\gamma{\cal D}c{\cal D}\bar c{\cal D}b\:,
\end{align}
where the numbers in the exponential of $\vvh$ represent the point-wise dimensions of Lie-algebra valued $(1,0)$-forms and $(0,0)$-forms, corresponding to $\gamma$, $c$, $\bar c$ and $b$ respectively. Note that the ghosts $c$ and $\bar c$ come with an factor of $2$, due to the field-redefinition \eqref{eq:fieldredef1}. The sign is also shifted as they are fermionic. 

Naively then, it seems the measure should be invarant under this re-scaling of fields. However, given our differential gauge-fixing constraint of Lorentz gauge, it is more natural to regulate the path-integral in momentum space instead. For example, we expand the measure of $\gamma$ as
\begin{equation}
    {\cal D}\gamma = \small\prod_{\lambda^i}\dd\gamma_i\:,
\end{equation}
where $\lambda^i$ are the eigenvalues of the kinetic operator associated to $\gamma$, and $\gamma_i$ denote the corresponding eigenvectors. Given the action \eqref{eq:S2b-redef}, we formally find
\begin{equation}
    {\cal D}\gamma = \small\prod_{\rm harm.}\dd\gamma_h\small\prod_{\lambda_{\tilde D_A}^i}\dd\gamma_i\small\prod_{\lambda_{\tilde D_A^\dagger}^i}\dd\gamma_i\:,
\end{equation}
where $\gamma_h$ are harmonic forms, and $\lambda_{\tilde D_A}^i$ and $\lambda_{\tilde D^\dagger_A}^i$ are formal eigenvalues of 
\begin{align}
    \tilde D_A&=e^{i\theta} D_A\notag\\
    \tilde D_A^\dagger&=e^{-i\theta} D_A^\dagger\:,
\end{align}
respectively, counted with multiplicity. Together with the bosonic $(3,0)$-form $b$, we see that
\begin{equation}
    {\cal D}\gamma{\cal D}b = \small\prod_{\rm harm.}\dd s_h\small\prod_{\lambda_{\tilde{\slashed D}^-_A}^i}\dd s_i\:,
\end{equation}
where $\tilde{\slashed D}^-_A=e^{i\theta}D_A+e^{-i\theta}D_A^\dagger$ is the Dirac operator on odd $(p,0)$-forms $s=\gamma+b$. The rescaling of fields then corrects the measure to\footnote{Note that $h_{\tilde D_A}^{(p,0)}(\Lg)=h_{D_A}^{(p,0)}(\Lg)$, as multiplying the operator by a phase does not change the cohomology.}
\begin{equation}
    {\cal D}\gamma{\cal D}b\rightarrow\vert\vh\vert^{-\tfrac12\left(h_{D_A}^{(1,0)}(\Lg)+h_{D_A}^{(3,0)}(\Lg)+\sum_i1\right)}{\cal D}\gamma{\cal D}b\:,
\end{equation}
where the sum is over non-zero eigenvalues of $\tilde{\slashed D}^-_A$, counted with multiplicity. Note that $\tilde{\slashed D}_A^-$ is self-adjoint, so its eigenvalues are real. 

Next, we notice that $\left(\tilde{\slashed D}_A^-\right)^2=\Delta_{D_A}$, where
\begin{equation}
    \Delta_{\del_A}=D_A D_A^\dagger+ D_A^\dagger D_A\:,
\end{equation}
is the Laplacian on $(1,0)+(3,0)$ forms. They hence have the same count for their eigenvalues, and the infinite sum can be regulated as
\begin{equation}
    \sum_i1\rightarrow\zeta_{\Delta_{ D_A}}(0)\:,
\end{equation}
where $\zeta_{\Delta_{ D_A}}(s)$ is the zeta function of $\Delta_{ D_A}$ which for ${\rm Re}(s)>>1$ can be expressed as
\begin{equation}
    \zeta_{{\Delta_{ D_A}}}(s)=\sum_i\frac{1}{\lambda_i^{2s}}\:.
\end{equation}
The $2$ in the exponential comes from the fact that for an eigenvalue $\lambda_i$ of $\tilde{\slashed D}_A^-$, the corresponding eigenvalue of $\Delta_{ D_A}$ is $\lambda_i^2$.

Including the contribution from the ghost fields, we thus get tat the measure ${\cal D}\Phi={\cal D}\gamma{\cal D}b{\cal D}c{\cal D}\bar c$ is re-scaled as
\begin{equation}
    {\cal D}\Phi\rightarrow\vvh^{-\tfrac12\left(h_{ D_A}^{1,0}(\Lg)+h_{ D_A}^{3,0}(\Lg)-2h^{0,0}_{ D_A}(\Lg)-2h^{3,0}_{ D_A}(\Lg)+\zeta_{\Delta^{(1,0)+(3,0)}_{ D_A}}(0)-2\zeta_{\Delta^{0,0}_{ D_A}}(0)-2\zeta_{\Delta^{3,0}_{ D_A}}(0)\right)}{\cal D}\Phi\:.
\end{equation}
Ala usual Serre-duality, it is straight forward to see that
\begin{equation}
    h_{ D_A}^{3,0}(\Lg)=h_{ D_A}^{0,0}(\Lg)\:,
\end{equation}
and $\Delta^{(3,0)}_{ D_A}$ has the same eigenvalues as $\Delta^{0,0}_{ D_A}$. Indeed, we have the following lemma
\begin{lemma}
\label{lem:Serre}
    For the operator $D_A$, we have the following Serre-duality result.
    \begin{equation}
        h^{p,0}_{D_A}(\Lg)\cong h^{3-p,0}_{D_A}(\Lg)\:,
    \end{equation}
    and $\Delta^{p,0}_{D_A}$ and $\Delta^{3-p,0}_{D_A}$ have isomorphic eigenspaces. 
\end{lemma}
\begin{proof}
    Let $\alpha\in\Omega^{p,0}(\Lg)$ be an eigenvector of $\Delta^{(p,0)}_{ D_A}$. That is
\begin{equation}
    (\Delta_{D_A}\alpha)^\mu=(D_AD_A^\dagger\alpha+D_A^\dagger D_A\alpha)^\mu=-{D_A}{\cal H}^{\mu\bar\nu}*{\bar D_A}*{\cal H}_{\bar\nu\rho}\alpha^\rho-{\cal H}^{\mu\bar\nu}*{\bar D_A}*{\cal H}_{\bar\nu\rho}D_A\alpha^\rho=\lambda\alpha^\mu\:,
\end{equation}
where $\{\mu,\nu,..\}$ denote Lie-algebra indices, and ${\cal H}$ is the chosen metric on $\Lg$. From now on, we supress Lie-algebra indecies and note that $\cal H$ is its own inverse on $\Lg$. We then have
\begin{equation}
    \Delta_{D_A}\alpha=-{D_A}{\cal H}*{\bar D_A}*{\cal H}\alpha-{\cal H}*{\bar D_A}*{\cal H}D_A\alpha=\lambda\alpha\:.
\end{equation}
Via some straight-forward manipulations, this can be rewritten as
\begin{equation}
    \Delta_{\bar D_A}(*{\cal H}\alpha)=-{\cal H}*{ D_A}*{\cal H}{\bar D_A}(*{\cal H}\alpha)-{\bar D_A}{\cal H}*D_A*{\cal H}({*\cal H}\alpha)=\lambda*{\cal H}\alpha\:.
\end{equation}
A further complex conjugation gives
\begin{equation}
    \Delta_{D_A}(*{\cal H}\bar\alpha)=\lambda*{\cal H}\bar\alpha\:.
\end{equation}
Contracting with $\Omega$, which commutes with derivatives as indeed $\Omega^{\bar a\bar b\bar c}$ is constant, we get
\begin{equation}
    \Delta_{D_A}\Omega\lrcorner(*{\cal H}\bar\alpha)=\lambda\Omega\lrcorner*{\cal H}\bar\alpha\:.\:.
\end{equation}
As the Lie-algebra is a self-dual space, and the operations we have done are invertible, it follows that $\Delta^{p,0}_{D_A}$ and $\Delta^{3-p,0}_{D_A}$ have isomorphic eigenspaces. The result follows.
\end{proof}

Using this lemma, we thus get
\begin{equation}
  {\cal D}\gamma{\cal D}b{\cal D}c{\cal D}\bar c\rightarrow\vvh^{-\tfrac12\left(h_{ D_A}^{1,0}(\Lg)-3h^{0,0}_{ D_A}(\Lg)+\zeta_{\Delta^{1,0}_{ D_A}}(0)-3\zeta_{\Delta^{0,0}_{ D_A}}(0)\right)}{\cal D}\gamma{\cal D}b{\cal D}c{\cal D}\bar c\:.  
\end{equation}
The partition function therefore becomes
\begin{equation}
    Z^{\rm{\small BRST}}_2=\vvh^{-\tfrac12\left(h_{ D_A}^{1,0}(\Lg)-3h^{0,0}_{ D_A}(\Lg)+\zeta_{\Delta^{1,0}_{ D_A}}(0)-3\zeta_{\Delta^{0,0}_{ D_A}}(0)\right)}\times Z^{\rm{\small BRST}}_{\rm hol\:CS}(D_A,\vvh,\theta)\:,
\end{equation}
where $Z^{\rm{\small BRST}}_{\rm hol\:CS}(D_A,\vvh,\theta)$ is the BRST partition function of the holomorphic Chern-Simons type theory \eqref{eq:S2b-redef}.

It is also convenient to remove the dependence of the phase $e^{i\theta}$ from the action \eqref{eq:S2b-redef}. By an additional rescaling
\begin{equation}
    \gamma\rightarrow e^{-i\tfrac12\theta}\gamma\:,\;\;\; b\rightarrow e^{i\tfrac32\theta}b\:,
\end{equation}
the action becomes
\begin{equation}
\label{eq:S2b-redef2}
    S_{2,\,{\rm gf}}=\int_X\tr\left(\gamma D_A\gamma+b D_A^\dagger\gamma+\bar c\, D_A^\dagger D_A c+\tfrac{1}{3\sqrt{\vh}}\gamma\,[\gamma,\gamma]+\tfrac{1}{\sqrt{\vh}}\bar c\, D_A^\dagger[\gamma,c]\right)\wedge\bar\Omega\:.
\end{equation}
Under this rescaling, the measure also rescales as
\begin{equation}
    {\cal D}\gamma{\cal D}b\rightarrow e^{-\tfrac{i}{2}\theta\left(h_{ D_A}^{1,0}(\Lg)-3h^{0,0}_{ D_A}(\Lg)+\zeta_{\Delta^{1,0}_{ D_A}}(0)-3\zeta_{\Delta^{0,0}_{ D_A}}(0)\right)}{\cal D}\gamma{\cal D}b\:.
\end{equation}
The partition function therefore becomes
\begin{equation}
    Z^{\rm{\small BRST}}_2=\vh^{-\tfrac12\left(h_{ D_A}^{1,0}(\Lg)-3h^{0,0}_{ D_A}(\Lg)+\zeta_{\Delta^{1,0}_{ D_A}}(0)-3\zeta_{\Delta^{0,0}_{ D_A}}(0)\right)}\times Z^{\rm{\small BRST}}_{\rm hol\:CS}(D_A,h)\:,
\end{equation}
where $Z^{\rm{\small BRST}}_{\rm hol\:CS}(D_A,h)$ is the BRST partition function of the anti-holomorphic Chern-Simons type theory \eqref{eq:S2b-redef2}. Note in particular that the explicit dependence of the partition function on $h$ is holomorphic. 

Being more explicit about the partition function for general $h$ is beyond the scope of this paper. If we however take the large complex structure limit, where $\vh\rightarrow\infty$, we get the much simpler expression
\begin{equation}
\label{eq:PartFuncBRSTlargeh0}
    \lim_{\vh\rightarrow\infty}Z^{\rm{\small BRST}}_2=\vh^{-\tfrac12\left(h_{ \del_A}^{1,0}(\Lg)-3h^{0,0}_{ \del_A}(\Lg)+\zeta_{\Delta^{1,0}_{ \del_A}}(0)-3\zeta_{\Delta^{0,0}_{ \del_A}}(0)\right)}\times Z^{\rm BRST}_{\rm 1-loop}(\del_A)\:, 
\end{equation}
where $Z^{\rm BRST}_{\rm 1-loop}(\del_A)$ is the one-loop BRST partition function of ordinary anti-holomorphic Chern--Simons theory with connection $\del_A$. This expression should be compared with \eqref{eq:FormalPartLargeh} of the formal quantisation in section \ref{sec:QuantFormal}.

The full $h$-dependence of the partition function is however still a bit deceptive. Though the dimensions of the cohomology groups $h_{\del_A}^{p,0}(\Lg)$ are independent of complex gauge, and thus a change of the holomorphic connection $\del_A$, this is not true of the zeta invariants $\zeta_{\Delta^{p,0}_{\del_A}}(0)$. However, in the case when $\del_A$ is hermitian with respect to a positive definite hermitian metric, with the corresponding metric used to define the Laplacians, we can use Lemma \ref{lem:Zeta} below for the linear combination of Hodge numbers and Zeta invariants to get
\begin{equation}
\label{eq:PartFuncBRSTLargeh}
    \lim_{\vh\rightarrow\infty}Z^{\rm{\small BRST}}_2=\vh^{-\frac{{\rm dim}(\Lg)}{480}\chi(X)}\times Z^{\rm BRST}_{\rm 1-loop}(\del_A)\:, 
\end{equation}
removing the implicit $h$-dependece of the exponential. However, we note further that the connection $\del_A$ depends on $h$ through the background requirement \eqref{eq:Instanton}. The partition function $Z^{\rm BRST}_{\rm 1-loop}(\del_A)$ therefore also depends on $h$ implicitly through $\del_A$. It is then desirable to find a scheme where this implicit $h$-dependence of the partition function can be removed. We will see that this can be done for the gauge theory on $\End_0(T^{1,0}X)$ described in section \ref{sec:EndTtheory}, by picking a convenient gauge. 

\subsection{A convenient gauge on ${\rm End}_0(T^{1,0}X)$, and a twisted theory}
\label{sec:Convgauge}
As mentioned, for the gauge theory on $\End_0(T^{1,0}X)$ described in section \ref{sec:EndTtheory}, we are restricting to the special directions where $h$ satisfies equation \eqref{eq:Inv=cc}, and the corresponding connection $\dd_{\tilde\nabla}$ on ${\rm End}_0(TX)$ is self-adjoint. This leads to a real gauge group $SU(3)$, with real generators. 

We will investigate what happens when we do a complex gauge transformation, switching from the background connection $\tilde\nabla$ to a different connection given by some metric $g$ on the tangent bundle. Let $V\in\Gamma(T^{1,0}X)$. Then
\begin{equation}
    \tilde\nabla_aV^b=h_{\bar c}^b\del_a\left(h_c^{\bar c}V^c\right)=h^{be}g_{\bar c e}\del_a\left(g^{d\bar c}h_{d c}V^c\right)=h^{be}\nabla_a\left(h_{e c}V^c\right)\:,
\end{equation}
where $\nabla$ is the Chern connection of $g$. Similarly, for $W\in\Gamma(\End(T^{1,0}X))$ we get
\begin{equation}
    \tilde\nabla_a{W^b}_c=h^{bd}h_{ce}\nabla_a\left(h^{eg}h_{df}{W^f}_g\right)\:.
\end{equation}
Note then that the operation
\begin{equation}
    {W^a}_b\rightarrow h^{ac}h_{bd}{W^d}_c
\end{equation}
squares to the identity, which will be important later. We will use short-hand notation 
\begin{equation}
    W\rightarrow HW
\end{equation}
to denote this transformation, which can also be thought of as a complex gauge transformation. 

If we then pick the gauge where we use the Kähler metric $g$ also on $\End(T^{1,0}X)$ in the definition of the inner product and adjoint, we find the following
\begin{proposition}
\label{prop:twist}
   For the specific directions in the deformation parameter space $Z^A$ where $\bar h_{a}^{\bar b}\propto h_a^{\bar b}$, that is deformations satisfying \eqref{eq:original}, the eigenvalue equation 
   \begin{equation}
    \Delta_{\del_{\tilde\nabla}}^{p,0}\alpha=\lambda\alpha
\end{equation}
becomes
\begin{equation}
\label{eq:RefinedEigen}
    H\circ\Delta_{\del_{\nabla}}^{p,0}(H\alpha)=\lambda\alpha\:.
\end{equation}
\end{proposition}
\begin{proof}
Given the inner product for $\alpha,\beta\in\Omega^{p,0}(\End(T^{1,0}X))$,
\begin{equation}
    (\alpha,\beta)=\int_X\alpha^a{}_b\wedge*\bar\beta^{\bar a}{}_{\bar b}\,g_{a\bar a}g^{b\bar b}\:,
\end{equation}
we can compute the adjoint of $\del_{\tilde\nabla}$ as
\begin{equation}
    (\del_{\tilde\nabla}^\dagger\alpha)^a{}_b=-{\bar h}_{bd}{\bar h}^{ac}*\delb*\left({\bar h}^{de}{\bar h}_{cf}\beta^f{}_e\right)\:.
\end{equation}
To get this expression, we have used the fact that $h_{ab}$ and similarly $h^{ab}$ and their complex conjugates are {\it symmetric}, i.e. they take values in the $\bar{\bf6}$ and $\bf 6$ irreducible representations of $SU(3)$ respectively. This follows as $h$ is chosen to be harmonic, which implies that $\Omega(h)$ is a primitive $(2,1)$-form (see Appendix \ref{app:spacial}), which is equivalent to $h_{ab}$ being symmetric.

We now use that we have chosen $h$ to be in a direction where $h_a{}^{\bar b}=c\,\bar h_a{}^{\bar b}$, where $c$ is real as noted above. Using this, we find that
\begin{equation}
    (\del_{\tilde\nabla}^\dagger\alpha)^a{}_b=-{h}_{bd}{h}^{ac}*\delb*\left({h}^{de}{h}_{cf}\beta^f{}_e\right)\:.
\end{equation}
That is, for these special directions in deformation space, we get
\begin{equation}
    \del_{\tilde\nabla}^\dagger\alpha=-H*\delb*(H\alpha)\:,
\end{equation}
where we also recognize $-*\delb*$ as the adjoint of the Chern connection, that is $\del_\nabla^\dagger=-*\delb*$. It follows that
\begin{align}
\Delta_{\del_{\tilde\nabla}}&=\del_{\tilde\nabla}^\dagger\del_{\tilde\nabla}+\del_{\tilde\nabla}\del_{\tilde\nabla}^\dagger\notag\\
&=H\circ\del^\dagger_\nabla\circ H\circ H\circ\del_\nabla\circ H+H\circ\del_\nabla\circ H\circ H\circ\del^\dagger_\nabla\circ H\notag\\
&=H\circ\Delta_{\del_{\nabla}}\circ H\:,
\end{align}
where the last equality follows as $H$ squares to the identity.
\end{proof}

Again, as $H$ squares to the identity, equation \eqref{eq:RefinedEigen} is then equivalent to
\begin{equation}
    \Delta_{\del_{\nabla}}^{p,q}(H\alpha)=\lambda(H\alpha)\:.
\end{equation}
It follow that in this gauge, for the specific directions in deformation space where \eqref{eq:original} is satisfied, $\Delta_{\del_{\tilde\nabla}}^{p,0}$ and $\Delta_{\del_{\nabla}}^{p,0}$ have the same eigenvalues, and hence the same Zeta invariants. Furthermore, with $\del_A = \del_{\tilde\nabla}$, in this gauge the quadratic gauged-fixed action of $Z^{\rm 1-loop}_{\rm hol\:CS}(\del_A)$ becomes
\begin{equation}
\label{eq:Hol-CS-Chern}
    S^{\rm 1-loop}_{2,\,{\rm gf}}=\int_X\tr\left(\gamma\del_\nabla\gamma+b\del_\nabla^\dagger\gamma+\bar c\,\del_\nabla^\dagger\del_\nabla c\right)\wedge\bar\Omega\:,
\end{equation}
where we have also performed the field redefinition $H\Phi\rightarrow\Phi$ on all the fields $\Phi=(\gamma,c,\bar c, b)$. As this transformation squares to the identity, its Jacobian can at most introduce an irrelevant sign to the partition function. 

In this gauge, the partition function then becomes (for large $\vh$)
\begin{equation}
\label{eq:Z-BRST-twist}
    Z^{\rm BRST}_2=\vh^{-\tfrac12\left(h_{\del_\nabla}^{1,0}(\Lg)-3h^{0,0}_{\del_\nabla}(\Lg)+\zeta_{\Delta^{1,0}_{\del_\nabla}}(0)-3\zeta_{\Delta^{0,0}_{\del_\nabla}}(0)\right)}\times Z^{\rm BRST}_{\rm 1-loop}(\nabla)\:,
\end{equation}
where $Z^{\rm BRST}_{\rm 1-loop}(\nabla)$ is the partition function of the anti-holomorphic Chern-Simons theory \eqref{eq:Hol-CS-Chern}, using the ordinary Chern connection. Using Lemma \ref{lem:Zeta} below, we see that the partition function becomes
\begin{equation}
    Z^{\rm BRST}_2=\vh^{-\tfrac{1}{60}\chi(X)}\times Z^{\rm BRST}_{\rm 1-loop}(\nabla)\:,
\end{equation}
where we have set ${\rm dim}(\Lg)=8$, the adjoint of $SU(3)$. Note in particular that the $h$-dependence is explicit, only through the determinant $\vh$. Though we stress again that we are restricted to choosing the special directions $h$ in deformation space. We shall refer to this gauge choice and the theory \eqref{eq:Z-BRST-twist} as the "twisted theory", due to the more gravitational nature of the theory. 

\subsection{The quantum theory: BV quantization}
Perhaps the most mathematically rigorous approach to quantizing the theory is using the BV quantization formalism. Ignoring zero modes, this formalism is equivalent to the BRST approach, as shown in the case of ordinary Chern--Simons theory in for example \cite{Axelrod:1991vq}. The BV-approach is however arguably the most rigorous approach to use when zero modes are turned on, and to deal with the corresponding volume factors of harmonic forms as we shall see below. For now, we turn off the zero modes, and consider only off-shell fluctuations.

The BV-approach is standard, and we will not be as detailed as in the BRST approach. The BV action results from promoting the field to a poly-form $\gamma\in\Omega^{\bullet,0}(\Lg)$, where odd degrees are bosonic and even degrees are fermionic. The resulting BV action can then be written neatly as
\begin{equation}
\label{eq:DefBVaction}
    S_{2,{\rm BV}}=\int_X\tr\left(\gamma{\cal D}^{\rm BV}_A\gamma+\tfrac13\gamma\{\gamma,\gamma\}\right)\wedge\bar\Omega=\vh\int_X\tr\left(\gamma D^{\rm BV}_A\gamma+\tfrac13\gamma[\gamma,\gamma]\right)\wedge\bar\Omega\:,
\end{equation}
with an implicit BV wedge-product between fields. That is, for $\alpha\in\Omega^{p,0}(\Lg)$ of fermionic degree $f_\alpha\in\{0,1\}$, where $\alpha$ has total degree $\vert\alpha\vert=p+f_\alpha$, and $\beta\in\Omega^{q,0}(\Lg)$ of fermionic degree $f_\beta$, we have
\begin{equation}
    \alpha\wedge^{\rm BV}\beta=(-1)^{pf_\beta}\alpha\wedge\beta\:.
\end{equation}
The kinetic operator also acts as
\begin{equation}
    D^{\rm BV}_A\alpha=(-1)^{f_\alpha}D_A\alpha\:.
\end{equation}
The bracket is also promoted to the BV bracket,
\begin{equation}
    [\alpha,\beta]=\alpha\wedge^{\rm BV}\beta-(-1)^{\vert\alpha\vert\vert\beta\vert}\beta\wedge^{\rm BV}\alpha\:.
\end{equation}
The promotion of the kinetic operator to the BV operator, and the wedge product to the BV wedge-product, ensures that ${\cal D}^{\rm BV}_A$ and $\{\,,\,\}$ (and also $D^{\rm BV}_A$ and $[\,,\,]$) together form a DGLA. This can also be elegantly expressed using the superspace formalism, see e.g. \cite{Axelrod:1991vq} for ordinary Chern--Simons, though we will not do so here.

One then needs to gauge-fix, or pick a Lagrangian submanifold in the space of fields. The common choice, as above, is to pick metrics on $X$ and $\Lg$, used to define adjoints, where in this gauge the partition function is simply
\begin{equation}
    Z^{\rm BV}_2=\int{\cal D}\gamma\big\vert_{\gamma\in{\rm ker}\left({\cal D}^\dagger_A\right)}e^{iS_{2,{\rm BV}}}\:.
\end{equation}
As above, we can re-scale our fields, though the re-slacing differs from the BRST approach. We now let
\begin{equation}
\label{eq:BVrescale}
    \gamma\rightarrow\tfrac{1}{\sqrt{\vh}}\gamma\:.
\end{equation}
This produces a Jacobian factor from the measure\footnote{Note that ${\rm ker}\left({\cal D}^\dagger_A\right)={\rm ker}\left(D^\dagger_A\right)$, as ${\cal D}^\dagger_A=\bar\vh D^\dagger_A$.}
\begin{equation}
    Z^{\rm BV}_2=\vh^{-\tfrac12\left(-3\zeta_{\Delta^{0,0}_{ D_A}}(0)+2\zeta_{\Delta^{1,0}_{ D_A}}(0)-\zeta_{\Delta^{2,0}_{ D_A}}(0)\right)}\int{\cal D}\gamma\big\vert_{\gamma\in{\rm ker}\left(D^\dagger_A\right)}e^{iS_{2,{\rm BV}}}\:,
\end{equation}
where now 
\begin{equation}
\label{eq:BVaction-rescaled}
    S_{2,{\rm BV}}=\int_X\tr\left(\gamma D^{\rm BV}_A\gamma+\tfrac{1}{3\sqrt{\vh}}\gamma[\gamma,\gamma]\right)\wedge\bar\Omega\:.
\end{equation}
Using Lemma \ref{lem:Serre}, the partition function then simplifies to
\begin{equation}
\label{eq:PartFuncBV}
    Z^{\rm BV}_2=\vh^{-\tfrac12\left(\zeta_{\Delta^{1,0}_{ D_A}}(0)-3\zeta_{\Delta^{0,0}_{ D_A}}(0)\right)}\int{\cal D}\gamma\big\vert_{\gamma\in{\rm ker}\left(D^\dagger_A\right)}e^{iS_{2,{\rm BV}}}\:,
\end{equation}
where the last factor is then the BV partition function of the re-scaled action \eqref{eq:BVaction-rescaled}. 

\subsubsection*{Zero modes and rigorous treatment of harmonic volumes}
We now come to include the volume factors of harmonic zero-modes appropriate for the BV formalism. One can "ignore" such zero-mode contributions, see for example \cite{Pestun:2005rp}, by not integrating over them in the path integral, or even pick and choose which zero-modes to integrate over in the definition of ones partition function. However, the zero-modes will play a role when we come to consider geometric anomalies in section \ref{sec:GeomAnomalies}, and they can be used to relate the zeta-invariants to more explicit geometric invariants, see Lemma \ref{lem:Zeta} below. We will therefore now spend some paragraphs to formalize the treatment of zero-modes in the BV approach. As it turns out, this will also serve to cure the discrepancy between the Hodge numbers in the exponential when comparing the BRST and BV approaches, see section \ref{sec:BV-Metric}.

At large values of the deformation parameter $h$, ignoring harmonic zero-modes, the partition function \eqref{eq:PartFuncBV} becomes
\begin{equation}
\label{eq:BVPartFuncLargeh}
    Z^{\rm BV}_2=\vh^{-\tfrac12\left(\zeta_{\Delta^{1,0}_{\del_A}}(0)-3\zeta_{\Delta^{0,0}_{\del_A}}(0)\right)}\int{\cal D}\gamma\big\vert_{\gamma\in{\rm Im}\left(\del^\dagger_A\right)}{\rm exp}\left(i\int_X\tr(\gamma\del_A^{\rm BV}\gamma)\wedge\bar\Omega\right)\:,
\end{equation}
as ${\rm ker}\left(\del^\dagger_A\right)={\rm Im}\left(\del^\dagger_A\right)$ when zero-modes are turned off. This should be multiplied by appropriate volume factors of harmonic forms when zero-modes are turned on. For-example, the volume of harmonic $(1,0)$-forms is formally
\begin{equation}
    \int\prod_{\rm harm.}\dd\gamma_h^{1,0}={\rm Vol}({\cal H}^{1,0}(\Lg))\:.
\end{equation}
The full volume factor of all the fields, ghosts, and anti-fields is then
\begin{equation}
    {\rm Vol}({\cal H})=\frac{{\rm Vol}({\cal H}^{1,0}(\Lg)){\rm Vol}({\cal H}^{3,0}(\Lg))}{{\rm Vol}({\cal H}^{0,0}(\Lg)){\rm Vol}({\cal H}^{2,0}(\Lg))}\:,
\end{equation}
where the volumes in the denominator come from Grassmann-valued fields.

To deal appropriately with these complex volume factors, let us first consider the absolute value of the partition function instead, which becomes
\begin{equation}
    \vert Z^{\rm BV}_2\vert^2=\vvh^{-\zeta_{\Delta^{1,0}_{\del_A}}(0)+3\zeta_{\Delta^{0,0}_{\del_A}}(0)}\left({\rm Vol}({\cal H}){\rm Vol}(\bar{\cal H})\right)^{\tfrac12}\,{\cal I}_{\rm RS}(\del_A)\:,
\end{equation}
where the square root is due to the fact that we are also integrating over a Lagrangian sub-manifold in the space of harmonic forms. The volume factors may then be combined into real volume factors
\begin{equation}
    {\rm Vol}(Z^p)={\rm Vol}({\cal H}^{p,0}(\Lg)){\rm Vol}({\cal H}^{0,p}(\bar\Lg))\:,
\end{equation}
where $Z^p$ is the real space
\begin{equation}
    Z^p={\cal H}^{p,0}(\Lg)\times{\cal H}^{0,p}(\bar\Lg)\:,
\end{equation}
where we note at this point we are not assuming that the Lie-algebra $\Lg$ is necessarily real.

At this point, the volume simply refers to integrating along harmonic directions. However, to properly perform this intgegral we need to equip $Z^p$ with a (hermitian) metric $g^p$, with corresponding hermitian two-form $\omega_{g^p}$, and use the corresponding covariant volume given by
\begin{equation}
    {\rm Vol}_{g^p}(Z^p)=\int_{Z^p}(*1)_{g^p}=\int_{Z^p}\frac{1}{(h^{p,0}(\Lg))!}\omega_{g^p}^{h^{p,0}(\Lg)}
\end{equation}
in the definition of the (absolute value of the) partition function. The partition function is then promoted to
\begin{equation}
    \vert Z^{\rm BV}_2\vert^2=\vvh^{-\zeta_{\Delta^{1,0}_{\del_A}}(0)+3\zeta_{\Delta^{0,0}_{\del_A}}(0)}{\rm Vol}_{g^\bullet}(Z^\bullet)^{\tfrac12}\,{\cal I}_{\rm RS}(\del_A)\:,
\end{equation}
where we use the short-hand notation
\begin{equation}
    {\rm Vol}_{g^\bullet}(Z^\bullet)=\frac{{\rm Vol}_{g^1}(Z^1){\rm Vol}_{g^3}(Z^3)}{{\rm Vol}_{g^0}(Z^0){\rm Vol}_{g^2}(Z^2)}\:.
\end{equation}
Note that using the covariant measure when we integrate over the harmonic part,\footnote{Note however that this measure is real, and this replacement can only be done for the definition of the absolute value (squared) of the partition function.} the partition function becomes insensitive to field-redefinitions of harmonic forms. Such field-redefinitions will hence not produce any Jacobian factors. However, the volumes do depend on what metric we choose on the space of harmonic forms. 

Restricting to the case of hermitian connections $\del_A$, corresponding to a hermitian metric ${\cal H}_{\mu\bar\nu}$, there is then a natural inner product on harmonic forms, given by 
\begin{equation}
\label{eq:IP0}
\langle\alpha,\beta\rangle_{g^p_0}=\int_X\tfrac{1}{(3-p)!}\omega^{3-p}\wedge\alpha^\mu\bar\beta^{\bar\nu}{\cal H}_{\mu\bar\nu}\:,
\end{equation}
where $\alpha$ and $\beta$ are in ${\cal H}^{p,0}(\Lg)$, and where $\omega$ is the Kähler form of the background metric on $X$.\footnote{The reader might worry that for the explicit theory defined for the special directions on ${\rm End}_0(T^{1,0}X)$, as noted in section \eqref{sec:EndTtheory}, the corresponding metric ${\cal H}$ is not positive definite, and the corresponding connections are only pseudo-hermitian. From a physics perspective such concerns are often swept under the rug. However, for the case at hand, it is arguably more rigorous, or perhaps more natural, to first twist the fields as in section \ref{sec:Convgauge}, before including the volume factors now using the metric induced by the ordinary Kähler metric $g$ on ${\rm End}_0(T^{1,0}X)$. We take this approach when we study geometric anomalies of the "twisted theory" in section \ref{sec:GeomAnomalies}.} Note that this inner-product is given in the canonical fields, before we perform the field redefinition \eqref{eq:BVrescale}. Arguably then, after doing the rescaling field redefinition \eqref{eq:BVrescale}, the natural metric on $Z^p$ which we denote by $g^p_\vvh$, in the rescaled fields corresponds to the inner-product
\begin{equation}
\label{eq:IPh}
\langle\alpha,\beta\rangle_{g^p_\vvh}=\frac{1}{\vvh}\int_X\tfrac{1}{(3-p)!}\omega^{3-p}\wedge\alpha^\mu\bar\beta^{\bar\nu}{\cal H}_{\mu\bar\nu}\:.
\end{equation}
For the case of hermitian connections especially, it is then natural to promote the (large $h$) partition function to
\begin{equation}
    \vert Z^{\rm BV}_2\vert^2=\vvh^{-\zeta_{\Delta^{1,0}_{\del_A}}(0)+3\zeta_{\Delta^{0,0}_{\del_A}}(0)}{\rm Vol}_{g_\vvh^\bullet}(Z^\bullet)^{\tfrac12}\,{\cal I}_{\rm RS}(\del_A)\:,
\end{equation}
where again we use the hermitian metric on the bundle in the definition of the adjoint when defining the holomorphic Ray--Singer torsion. 

Note however that
\begin{equation}
    {\rm Vol}_{g_h^\bullet}(Z^\bullet)=\vvh^{h^{0,0}(\Lg)-\,h^{1,0}(\Lg)+h^{2,0}(\Lg)-h^{0,0}(\Lg)}{\rm Vol}_{g_0^\bullet}(Z^\bullet)\:,
\end{equation}
where $g_0^p$ is the metric using canonical fields. After using Serre-duality, $h^{p,0}(\Lg)=h^{3-p,0}(\Lg)$, the partition function therefore reads
\begin{equation}
    \vert Z^{\rm BV}_2\vert^2=\vvh^{-\zeta_{\Delta^{1,0}_{\del_A}}(0)+3\zeta_{\Delta^{0,0}_{\del_A}}(0)}{\rm Vol}_{g_0^\bullet}(Z^\bullet)^{\tfrac12}\,{\cal I}_{\rm RS}(\del_A)\:,
\end{equation}
where the last two factors is simply the absolute value square of the BV partition function of ordinary anti-holomorphic Chern--Simons theory. 

We remark on the discrepency of Hodge numbers appearing in the exponential as compared with the BRST result \eqref{eq:PartFuncBRSTlargeh0}. The difference between the Hodge numbers appearing in the exponential as compared to the BRST approach, see equation \eqref{eq:PartFuncBRSTlargeh0}, discerns from the type of zero-modes included in the path integral, and the slightly different re-scalings of fields we have done in the two different approaches. We can use a similar covariant treatment of volumes of zero-modes in the BRST formalim. After the various field redefinitions, the absolute value of the BRST partition function becomes
\begin{equation}
    \vert Z^{\rm BRST}_2\vert^2=\vvh^{-\zeta_{\Delta^{1,0}_{\del_A}}(0)+3\zeta_{\Delta^{0,0}_{\del_A}}(0)}\frac{{\rm Vol}_{g_\vvh^1}(Z^1){\rm Vol}_{g_\vvh^3}(Z^3)}{{\rm Vol}_{g_{\vvh^2}^0}(Z^0){\rm Vol}_{g_{\vvh^2}^3}(Z^3)}\,{\cal I}_{\rm RS}(\del_A)\:,
\end{equation}
where the different scalings of $\vvh$ of the metrics on the $Z^p$'s are due to the different field redefinitions of bosonic and fermionic fields \eqref{eq:fieldredef1}. It follows that
\begin{equation}
    \vert Z^{\rm BRST}_2\vert^2=\vvh^{-\zeta_{\Delta^{1,0}_{\del_A}}(0)+3\zeta_{\Delta^{0,0}_{\del_A}}(0)-h^{1,0}(\Lg)+3h^{0,0}(\Lg)}\frac{{\rm Vol}_{g_0^1}(Z^1)}{{\rm Vol}_{g_0^0}(Z^0)}\,{\cal I}_{\rm RS}(\del_A)\:.
\end{equation}
For hermitian connections (with hermitian metrics), this becomes using Lemma \ref{lem:Zeta} below
\begin{equation}
    \vert Z^{\rm BRST}_2\vert^2=\vvh^{-\tfrac{{\rm dim}(\Lg)}{240}\chi(X)}\frac{{\rm Vol}_{g_0^1}(Z^1)}{{\rm Vol}_{g_0^0}(Z^0)}\,{\cal I}_{\rm RS}(\del_A)\:.
\end{equation}
It is interesting to note that it is precisely the combination of terms which appear in the exponential in the BRST formalism which turns out to be proportional to the Euler number $\chi(X)$. 

However, it is arguably the BV formalism that gives the most rigorous treatment of the harmonic zero-modes. In particular, it is the combination 
\begin{equation}
    {\rm Vol}_{g_0^\bullet}(Z^\bullet)^{\tfrac12}\,{\cal I}_{\rm RS}(\del_A)=\left(\frac{{\rm Vol}_{g_0^1}(Z^1){\rm Vol}_{g_0^3}(Z^3)}{{\rm Vol}_{g_0^0}(Z^0){\rm Vol}_{g_0^2}(Z^2)}\right)^{\tfrac12}\,{\cal I}_{\rm RS}(\del_A)
\end{equation}
which enjoys particularly nice anomaly formulas when varying the background metrics and complex structure, which we will return to in section \ref{sec:GeomAnomalies} when we come to study geometric anomalies in more detail.

Let us now say more about the Zeta-invariants appearing in the exponential. Indeed, using the anomaly formulas derived by Bismut etal \cite{bismut1988analytic1, bismut1988analytic2, bismut1988analytic3}, derived for holomorphic vector bundles with hermitian connections, we can relate the combination of Zeta-invariants to topological invariants. We have the following lemma.
\begin{lemma}
\label{lem:Zeta}
    Let $X$ be a Calabi--Yau three-fold, with Kähler metric $g$, and let $V\rightarrow  X$ be a holomorphic vector bundle with a hermitian connection $\del_A$ corresponding to a hermitian metric on $V$. Let $\Delta_{\del_A}^{p,0}$ denote the corresponding Laplacian on $\Omega^{p,0}(\Lg)$, where $\Lg$ is the Lie-algebra of the corresponding gauge theory.  We assume that $\Lg$ is traceless. We then have
    \begin{equation}
        h^{1,0}(\Lg)-3h^{0,0}(\Lg)+\zeta_{\Delta^{1,0}_{\del_A}}(0)-3\zeta_{\Delta^{0,0}_{\del_A}}(0)=\frac{{\rm dim}(\Lg)}{120}\int_X{\rm ch}_3(X)=\frac{{\rm dim}(\Lg)}{240}\chi(X)\:,
    \end{equation}
    where $\chi(X)$ is the Euler number of $X$, and 
    \begin{equation}
        {\rm ch}_3(X)=\tfrac{1}{3!}\tr\left(\Big(\tfrac{i}{2\pi}R\Big)^3\right)=\tfrac12e(X)
    \end{equation}
    is the third Chern character, where $e(X)=c_3(X)$ is the Euler class of the Calabi--Yau $X$. 
\end{lemma}
\begin{proof}
    The proof relies on the observation that under a re-scaling of the metric $g$ on $X$,
    \begin{equation}
        g\rightarrow\frac{1}{a}g\:,
    \end{equation}
    where $a$ is a constant on $X$, the Laplacians scale as
    \begin{equation}
        \Delta_{\del_A}^{p,0}\rightarrow a\Delta_{\del_A}^{p,0}\:.
    \end{equation}
    The dependence of the of the harmonic Volume factors and holomorphic Ray--Singer torsion on $a$ is therefore
    \begin{equation}
    \label{eq:Ia}
        \del_a\log\left({\rm Vol}_{g^\bullet}(Z^\bullet)^{\tfrac12}\,{\cal I}_{\rm RS}(\del_A)\right)=-\frac{1}{2a}\left(h^{1,0}(\Lg)-3h^{0,0}(\Lg)+\zeta_{\Delta^{1,0}_{\del_A}}(0)-3\zeta_{\Delta^{0,0}_{\del_A}}(0)\right)\:,
    \end{equation}
    where we have used the definition \eqref{eq:IPh} of the metric $g^p$, and the Serre-duality results $h^{p,0}(\Lg)=h^{3-p,0}(\Lg)$ and $\zeta_{\Delta^{p,0}_{\del_A}}(0)=\zeta_{\Delta^{3-p,0}_{\del_A}}(0)$.
       
    However, we can also work out the change of ${\cal I}_{\rm RS}(\del_A)$ using the metric anomlaly formulas of Bismut et al \cite{bismut1988analytic1, bismut1988analytic2, bismut1988analytic3}. The result is
    \begin{equation}
    \label{eq:IaBismut}
        \del_a\log\left({\rm Vol}_{g^\bullet}(Z^\bullet)^{\tfrac12}\,{\cal I}_{\rm RS}(\del_A)\right)=-\frac{1}{2a}\cdot\frac{{\rm dim}(\Lg)}{120}\int_X{\rm ch}_3(X)\:.
    \end{equation}
    Comparing equations \eqref{eq:Ia} and \eqref{eq:IaBismut} gives the result. 
\end{proof}
Using this lemma, the (absolute value of the) large $h$ partition function in the BV formalism \eqref{eq:BVPartFuncLargeh} then becomes (again assuming a hermitian connection)
\begin{equation}
\label{eq:AbsPartFuncBVLargeh}
    \vert Z^{\rm BV}_2\vert^2=\vvh^{-\frac{{\rm dim}(\Lg)}{240}\chi(X)+h^{1,0}(\Lg)-3h^{0,0}(\Lg)}{\rm Vol}_{g_0^\bullet}(Z^\bullet)^{\tfrac12}\,{\cal I}_{\rm RS}(\del_A)\:.
\end{equation}
Of course, this is the absolute value (square) of the partition function. Indeed, in taking the absolute value we where able to define a real measure on the space of zero-modes, which we could use to define a real volume. The full partition function also includes a phase. Wether or not such a phase can be defined in general comes down to anomalies, which will be given a broader treatment in section \ref{sec:anomalies}. However, as the last two factors in \eqref{eq:AbsPartFuncBVLargeh} simply give the absolute value square of the one-loop partition function of ordinary (anti-)holomorphic Chern--Simons theory, it is tempting and reasonable to simply define the BV partition function as
\begin{equation}
\label{eq:BVPartFuncLargehHerm}
    Z^{\rm BV}_2=\vh^{-\frac{{\rm dim}(\Lg)}{480}\chi(X)+\tfrac12h^{1,0}(\Lg)-\tfrac32h^{0,0}(\Lg)}\int{\cal D}\gamma\big\vert_{\gamma\in{\rm ker}\left(\del^\dagger_A\right)}{\rm exp}\left(i\int_X\tr(\gamma\del_A^{\rm BV}\gamma)\wedge\bar\Omega\right)\:,
\end{equation}
where the path integral is simply the one-loop partition function of ordinary anti-holomorphic Chern--Simons theory in the BV approach.

\subsection{Alternative viewpoint: metric re-scaling}
\label{sec:BV-Metric}
The appearance of the Hodge numbers in the exponential of \eqref{eq:BVPartFuncLargehHerm} are a bit unsatisfying, as compared to what we got in the BRST approach. We now present an alternative viewpoint, where instead of doing a field redefinition, we interpret the theory as a new theory with a re-scaled background metric on $X$ instead. 

Let's recall the deformed BV theory \eqref{eq:DefBVaction}
\begin{equation}
    S_{2,{\rm BV}}=\vh\int_X\tr\left(\gamma D^{\rm BV}_A\gamma+\tfrac13\gamma[\gamma,\gamma]\right)\wedge\bar\Omega\:.
\end{equation}
At large $h$, only the quadratic fluctuations will be relevant for the path integral.\footnote{This is akin to sending $\hbar\rightarrow0$ in normal QFT, and considering the one-loop partition function of
\begin{equation}
 Z=\int{\cal D}\phi\,\exp\left(\tfrac{i}{\hbar}S(\phi)\right)\:.   
\end{equation}
}
The relevant large $h$ parth integral to consider is therefore
\begin{equation}
    \lim_{\vh\rightarrow\infty}Z^{\rm BV}_2=\int{\cal D}\gamma\big\vert_{\gamma\in{\rm ker}\left(\del^\dagger_A\right)}{\rm exp}\left(i\int_X\tr(\gamma\vh\del_A^{\rm BV}\gamma)\wedge\bar\Omega\right)\:.
\end{equation}
With the gauge-fixing constraint $\gamma\in{\rm ker}\left(\del^\dagger_A\right)$, the kinetic operator is essentially a Dirac operator, whose square is proportional to the Laplacian
\begin{equation}
    \left(\vh\del_A+\bar\vh\del_A^\dagger\right)^2=\vvh^2\Delta_{\del_A}\:.
\end{equation}
The point is that this is the same as the ordinary Laplacian, but with the re-scaled metric
\begin{equation}
    \tilde g=\frac{1}{\vvh^2}g\:.
\end{equation}
Ignoring zero-modes, the absolute value square of the partition function then becomes
\begin{equation}
    \vert Z^{\rm BV}_2\vert^2={\cal I}_{\rm RS}(\del_A,\tilde g)\:,
\end{equation}
where we use the metric $\tilde g$ in the definition of the holomorphic Ray--Singer torsion ${\cal I}_{\rm RS}(\del_A,\tilde g)$.

From this point of view, when including volume forms, the (absolute value) of the partition function should be the same as the ordinary partition function, but using the re-scaled metric $\tilde g$ instead. That is
\begin{equation}
    \vert Z^{\rm BV}_2\vert^2={\rm Vol}_{\tilde g_0^\bullet}(Z^\bullet)^{\tfrac12}\,{\cal I}_{\rm RS}(\del_A,\tilde g)\:,
\end{equation}
where we use $\tilde g$ for the inner product $\tilde g_0^p$ on ${\cal H}^{p,0}(\Lg)$, which in canonical fields $\alpha$ and $\beta$ in ${\cal H}^{p,0}(\Lg)$ reads
\begin{align}
\langle\alpha,\beta\rangle_{\tilde g^p_0}&=\int_X\tfrac{1}{(3-p)!}\tilde\omega^{3-p}\wedge\alpha^\mu\bar\beta^{\bar\nu}{\cal H}_{\mu\bar\nu}\notag\\
&=\frac{1}{\vvh^{2(3-p)}}\int_X\tfrac{1}{(3-p)!}\omega^{3-p}\wedge\alpha^\mu\bar\beta^{\bar\nu}{\cal H}_{\mu\bar\nu}=\frac{1}{\vvh^{2(3-p)}}\langle\alpha,\beta\rangle_{g^p_0}\:.
\end{align}
The absolute value of the partition function then becomes
\begin{equation}
    \vert Z^{\rm BV}_2\vert^2=\vvh^{-\zeta_{\Delta^{1,0}_{\del_A}}(0)+3\zeta_{\Delta^{0,0}_{\del_A}}(0)-h^{1,0}(\Lg)+3h^{0,0}(\Lg)}{\rm Vol}_{g_0^\bullet}(Z^\bullet)^{\tfrac12}\,{\cal I}_{\rm RS}(\del_A,g)\:,
\end{equation}
where ${\cal I}_{\rm RS}(\del_A,g)$ is now the ordinary holomorphic Ray--Singer torsion computed using $g$. For hermitian connections, using Lemma \ref{lem:Zeta}, this is then simply 
\begin{equation}
    \vert Z^{\rm BV}_2\vert^2=\vvh^{-\tfrac{{\rm dim}(\Lg)}{240}\chi(X)}{\rm Vol}_{g_0^\bullet}(Z^\bullet)^{\tfrac12}\,{\cal I}_{\rm RS}(\del_A,g)\:,
\end{equation}
where the last two terms are the absolute value square of the one-loop BV partition function of ordinary anti-holomorphic Chern--Simons theory. 

From this viewpoint, it is then tempting to simply define the large $\vh$ BV partition function (for hermitian connections) to be
\begin{equation}
\label{eq:PartFuncBVMetric}
    Z_2^{\rm BV}=\vh^{-\tfrac{{\rm dim}(\Lg)}{480}\chi(X)}\int{\cal D}\gamma\big\vert_{\gamma\in{\rm ker}\left(\del^\dagger_A\right)}{\rm exp}\left(i\int_X\tr(\gamma\del_A^{\rm BV}\gamma)\wedge\bar\Omega\right)\:.
\end{equation}
The difference between this and \eqref{eq:BVPartFuncLargehHerm}, is that we are now also scaling the metrics we use for computing the volume factors of harmonic forms. However, given the discussion above, this seems like a natural thing to do. Note that the result then also agrees more with the, perhaps more naive, BRST result \eqref{eq:PartFuncBRSTLargeh}. To keep things simple, we will take this as our large $h$ partition function for the remainder of the paper, where we now come to study anomalies.

\section{Anomalies and geometric dependence}
\label{sec:anomalies}
Having computed the partition function, we now move to study its anomalies. There are two types of anomalies to consider. The first is potential gauge and gravitational anomalies, where the partition function is not gauge invariant even though the classical theory is. The second type of anomalies is the dependence of the partition function on background geometric structures. Famous examples of this include the holomorphic anomaly equation of the type IIB topological string \cite{Bershadsky:1993cx, Bershadsky:1993ta}, or the metric dependence of the partition function due to the choice of gauge fixing \cite{bismut1988analytic1, bismut1988analytic2, bismut1988analytic3, Pestun:2005rp}. We will study both types of anomalies, focusing on the metric dependence and holomorphic dependence of the background complex structure of the partition function for the geometric anomalies.

In the current section, unless other wise stated, we also restrict to assuming real Lie-algebras where the background connection is hermitian, so the gauge group is real and compact, and the usual anomaly formulas, in particular the results of Bismut etal \cite{bismut1988analytic1, bismut1988analytic2, bismut1988analytic3}, apply.\footnote{Again, this means that in the case of the ${\rm End}_0(T^{1,0}X)$-theory, we restrict to background connections and complex structure deformations corresponding to the special directions satisfying \eqref{eq:Inv=cc}. Though as noted above, the corresponding metric is non-positive definite, and the background connection is only pseudo-hermitian, a concern often glossed over in physics.} We will also work in the large $\vh$ limit where the theory becomes quadratic, leaving the study of higher loop anomalies in the generically deformed theory \eqref{eq:DefBVaction}, where the kinetic operator is no longer given by a connection, to future work. We assume that we have quantized the theory using the BV formalism, so the results of Bismut etal and anomaly formulas of holomorphic Chern--Simons theory more readily apply. Finally, we remark that most of this section is simply the application of well known results to the theory at hand. 

\subsection{Gauge and gravitational anomalies}
It is well known that generic holomorphic Chern--Simons theories suffers from gauge anomalies. Anomalies are associated to an anomalous phase transformation of the partition function, and so do not appear in the absolute value $\vert Z\vert$. Such anomalies can occur for global symmetries, and local (gauge and gravitational) symmetries. Of these, only anomalies in local symmetries are “fatal”, indicating that the theory is inconsistent. We compute these anomalies in the coming subsections, and then suggest a number of ways to cancel them, focusing in particular on the theories described above where the gauge bundle is $\End(T^{1,0}X)$.

From a geometric perspective, the partition function $Z$ should be interpreted as a section of a certain holomorphic determinant line bundle over the complex configuration space $\cal C$ of geometric structures on $X$. If the Chern connection on this bundle has non-zero curvature, ${\cal F}_{\rm Det}\neq 0$, then the phase of $Z$ is not fully specified by the gauge-invariant data of a point in the configuration space. This implies there is an anomaly. 

Let's begin by considering the anomalies of ordinary holomorphic Chern--Simons theory, and then adapt the results to the case at hand. Applying the results of \cite{bismut1988analytic1, bismut1988analytic2, bismut1988analytic3, Bittleston:2022nfr} to holomorphic Chern--Simons theories on a bundle $\End(E)$, the curvature reads
\begin{equation}
    {\cal F}_{\rm Det}=\frac{1}{2}\left[\int_X{\rm td}(X)\wedge{\rm ch}(\End(E))\right]_{(1,1)}:=\left[\int_X{\cal P}\right]_{(1,1)}\:,
\end{equation}
where, somewhat heuristically, the characteristic polynomials ${\rm td}(X)$ and ${\rm ch}(\End(T^{1,0}X))$ are constructed using curvatures on the “universal geometry”, where we think of the manifold with a given geometric configuration as a fiber in the fibration $X\rightarrow{\cal C}$. The integration over $X$
should then be read as an integration over the fibres, resulting in a $(1,1)$-form on $\cal C$. As $X$ is six-dimensional, only the $(4,4)$-component (as a form on the total space) of the curvature polynomial $\cal P$ contributes to ${\cal F}_{\rm Det}$.

For a generic holomorphic Chern--Simons theory on a bundle $E$ over a Calabi--Yau three-fold $X$, one finds the anomaly polynomaial can be written as
\begin{equation}
    \label{eq:AnomalyPolHolCS}
    {\cal P}=\frac12{\rm ch}_4(\End(E))+\frac{{\rm dim}(\Lg)}{40}{\rm ch}_4(X)-\frac{1}{24}{\rm ch}_2(X)\wedge{\rm ch}_2(\End(E))\:,
\end{equation}
where ${\rm dim}(\Lg)$ is the dimension of the Lie-algebra $\Lg$. The first term is the pure gauge anomaly for holomorphic Chern--Simons theory, The second term is the gravitational anomaly associated to local coordinate changes that preserve $\Omega$, and the final term is a mixed gauge-gravitational anomaly. The last two terms of the anomaly polynomial are often ignored when dealing with holomorphic Chern--Simons theory as a gauge theory, particularly on flat geometries. Indeed, if we are not concerned with gravitational anomalies, the second term is not relevant. The third term may however be useful when canceling the gauge anomaly on curved backgrounds. Let us see why.

Firstly, when canceling anomalies, it is useful to consider Lie-algebras satisfying certain factorization properties. For the case at hand, we consider gauge groups where
\begin{equation}
    {\rm ch}_4(\End(E))=C_{\Lg}{\rm ch}_2(\End(E))^2\:,
\end{equation}
where $C_{\Lg}$ is an $\Lg$-dependent constant. Ignoring the gravitational parts of the anomaly, the polynomial then reads
\begin{equation}
    {\cal P}=\frac{C_{\Lg}}{2}{\rm ch}_2(\End(E))^2\propto\tr(F\wedge F)^2\:.
\end{equation}
Via the descent procedure, this gives an anomaly proportiaonal to
\begin{equation}
    \int_X\tr(F\wedge F)\tr(F\epsilon)\:,
\end{equation}
where $\epsilon\in\Omega^0(\Lg)$ is the infinitesimal gauge transformation of the background.\footnote{In the background field expansion and BV formalism, the gauge transformations of the fluctuations are fixed when the path integral is performed. However, in the background field expansion, the classical (master) action is still invariant under gauge transformations of the background connection, and such gauge transformations may lead to anomalies of the partition function.} In (holomorphic) supergravity theoies, such an anomaly may be canceled with local counter-terms via a Green--Schwarz mechanism \cite{Green:1984sg}, see e.g. \cite{Costello:2015xsa, Costello:2016mgj, Costello:2019jsy, Costello:2021kiv, Bittleston:2022nfr, Ashmore:2023vji, Ashmore:2025fxr}. If we however only have the gauge field at our disposal, to cancel the anomaly with a local and globally defined counter-term, it seems necessary to impose the to impose a topological constraint, such that
\begin{equation}
\label{eq:topconstr1}
    \left[\tr(F\wedge F)\right]=0\:,
\end{equation}
in cohomology. That is,
\begin{equation}
\label{eq:Triv1}
    \tr(F\wedge F)=\delb\tau\:,
\end{equation}
for some gauge invariant $(2,1)$-form $\tau$. The anomaly may then be canceled by adding a globally well-defined counter term proportional to
\begin{equation}
    \int_X\tau\wedge\tr(F\wedge\alpha_0)\:,
\end{equation}
where $\alpha_0\in\Omega^{0,1}(\End(E))$ is the dynamical part of the background connection of holomorphic Chern--Simons theory, whose background gauge transformation is $\delta\alpha_0=\delb\epsilon$.\footnote{Specifically, we parametrise the background connection $A$ on $\End(E)$ as
\begin{equation}
    A=A_0+\alpha_0\:,
\end{equation}
where $A_0$ is an in in general non-global fixed background connection, and $\alpha_0$ is a small deformation of this background, i.e. $\delb\alpha_0=0$. The background gauge transformation is then $\delta A=\delta\alpha_0=\delb\epsilon$. Note also that anomalies and counter terms are intrinsically a one-loop (or higher) effect. Modulo higher loop effects, we therefore need only consider classical fields in the counter terms.\label{foot:Background}}

The constraint \eqref{eq:Triv1} is however rather stringent. For example, let's consider backgrounds of (poly-)stable bundles, where the Donaldson--Uhlenbeck--Yau theorem \cite{0529.53018, 0615.58045} implies the existence of a unique hermitian Yang--Mills connection $A$ on $\End(E)$, whose curvature then satisfies the anti-self-dual constraint
\begin{equation}
\label{eq:ASD}
    *F=-\omega\wedge F\:,
\end{equation}
where $\omega$ is the Kähler form of the Calabi--Yau. Let $\cal H$ denote the corresponding hermitian Yang--Mills metric. Computing the norm of $F$ using $\cal H$, one finds
\begin{equation}
    \vert\vert F\vert\vert^2=\int_X\tr_{\cal H}(F\wedge*\bar F)=\int_XF^\mu{}_{\nu}\wedge*\bar F^{\bar\mu}{}_{\bar\nu}{\cal H}^{\nu\bar\nu}{\cal H}_{\mu\bar\mu}\:,
\end{equation}
where $F^\mu{}_{\nu}=\delb({\cal H}^{\mu\bar\kappa}\del{\cal H}_{\bar\kappa\nu})$, and where $\{\mu,\nu,..\}$ are bundle indices in the fundamental representation. A short computation then reveals that
\begin{equation}
    \bar F^{\bar\mu}{}_{\bar\nu}{\cal H}^{\nu\bar\nu}{\cal H}_{\mu\bar\mu}=-F^{\nu}{}_{\mu}\:.
\end{equation}
One therefore finds that
\begin{equation}
    \vert\vert F\vert\vert^2=-\int_X\tr(F\wedge*F)=\int_X\omega\wedge\tr(F\wedge F)=0\:,
\end{equation}
where the last equality follows from \eqref{eq:ASD} and \eqref{eq:Triv1}. It follows that $F=0$, and the resulting bundles are flat.


However, including the third "mixed anomaly" term of \eqref{eq:AnomalyPolHolCS} on curved backgrounds, the decent procedure gives an anomaly proportional to
\begin{equation}
    \int_X\left(C_{\Lg}\,{\rm ch}_2(\End(E))-\frac{1}{12}{\rm ch}_2(X)\right)\wedge\tr(F\epsilon)\:.
\end{equation}
Imposing the more non-trivial topological constraint
\begin{equation}
    C_{\Lg}\,{\rm ch}_2(\End(E))=\frac{1}{12}{\rm ch}_2(X)
\end{equation}
in cohomology, the gauge anomaly may then be cancelled by a similar procedure to above. 

Let us now consider the anti-holomorphic Chern--Simons theory at hand. Using the BV approach, at large values of the deformation parameter $h$, the classical BV action of interest reads
\begin{equation}
    S_{2,{\rm BV}}=\int_X\tr(\gamma \del^{\rm BV}_A\gamma)\wedge\bar\Omega\:.
\end{equation}
As the connection is assumed to be hermitian, given by some hermitian metric $\cal H$, we get\footnote{Recall that the bundle is $\End(E)$, so ${\cal H}^{-1}={\cal H}$.}
\begin{equation}
    \del_A={\cal H}\circ\del\circ{\cal H}\:.
\end{equation}
We can then do a field redefinition \begin{equation}
\label{eq:FieldRedefH}
    {\cal H}\gamma\rightarrow\gamma\in\Omega^{\bullet,0}(\End({\bar E}))\:,
\end{equation}
where the classical action simply becomes the complex conjugate of ordinary holomorphic Chern--Simons theory. Viewing the partition function as a section of a holomorphic determinant line bundle as above, the corresponding curvature ${\cal F}_{\rm Det}$ therefore becomes the negative of ordinary holomorphic Chern--Simons theory. 

It should be noted that in doing the field redefinition \eqref{eq:FieldRedefH} we do have to worry about a potential Jacobian factor from the path integral measure. However, as $\cal H$ is a real operator its determinant will not affect the phase of the partition function and therefore not the anomaly either. Indeed, as ${\cal H}^2=1$, this Jacobian at most gives a sign factor to the partition function, which we will ignore.\footnote{This is perhaps a bit dubious as $\cal H$ maps between different bundles, either $\End(E)$ to $\End({\bar E})$ or vice versa. However, the Jacobien for such an operator is defined as 
\begin{equation}
    {\rm Jac}({\cal H})={\rm det}({\cal H}):=\sqrt{\det({\cal H}^\dagger{\cal H})}=1\:,
\end{equation} 
as as ${\cal H}^\dagger=\cal H$ and ${\cal H}^2=1$.} 
From hereon, when considering gauge and gravitational anomalies, we can therefore assume that we are working with an ordinary holomorphic Chern--Simons theory on a holomorphic bundle $\End(E)$.

\subsection{Anomalies of the ${\rm End}_0(T^{1,0}X)$-theory}
We now consider anomalies of the gauge theory of section \ref{sec:EndTtheory}, again picking one of the special directions $h$ where the background connection is self-adjoint and pseudo-hermitian. We view this theory as a {\it gauge theory} on $\End(T^{1,0}X)$, equipped with a background pseudo-hermitian connection solving the instanton equation of motion \eqref{eq:Instanton}. To apply the usual gauge anomaly formulas, where the gauge curvature is given by the curvature $\tilde R$ of $\tilde\nabla$, it is then also important that when we quantize we gauge-fix using the $\tilde\nabla$-compatible metric on $\End(T^{1,0}X)$ derived from $h$, given by \eqref{eq:HermMetr}, instead of using the gauge-fixing choice using the background metric of section \ref{sec:Convgauge} which leads to the twisted theory. However, as this metric is not-positive definite, from a mathematical perspective this sub-section should perhaps be taken with a grain of salt. Such issues of non-positive definiteness are however often ignored when it comes to physics.\footnote{For example, in the usual Lorenz gauge fixing condition of particle physics one uses the non-positive definite Minkowski metric.} 

Decomposing under $SU(3)$, we have
\begin{equation}
    \End(T^{1,0}X)={\bf 3}\otimes{\bf\bar 3}={\bf 8}+{\bf 1}\:.    
\end{equation}
As mentioned, the connection symbols are trace-free by Proposition \ref{prop:TorsionFree}, and they hence take values in the Lie-algebra, or $\bf 8$ irreducible representation of $SU(3)$, i.e. we can treat it as an $SU(3)$ gauge theory on the background bundle ${\rm End}_0(T^{1,0}X)$, where zero denotes the trace-free part.  

The corresponding anomaly polynomial is
\begin{equation}
    {\cal P}=\frac12{\rm ch}_4({\bf 8})+\frac{8}{40}{\rm ch}_4(X)-\frac{1}{24}{\rm ch}_2(X)\wedge{\rm ch}_2({\bf 8})\:.
\end{equation}
For $SU(3)$, we have
\begin{equation}
    {\rm ch_4}({\bf 8})=\frac{1}{24}({\rm ch}_2({\bf 8}))^2\:,
\end{equation}
leading to the anomaly polynomial 
\begin{equation}
    {\cal P}=\frac{1}{24}\left(\frac12{\rm ch}_2({\bf 8})-{\rm ch}_2(X)\right)\wedge{\rm ch}_2({\bf 8})+\frac{1}{5}{\rm ch}_4(X)\:.
\end{equation}
Unfortunately, in the case of $SU(3)$ we also have
\begin{equation}
    {\rm ch}_2({\bf 8})=6\,{\rm ch}_2({\bf 3})\:,
\end{equation}
leading to the anomaly polynomial
\begin{equation}
\label{eq:anomalyPolSU3}
    {\cal P}=\frac{1}{24}\left(3\,{\rm ch}_2({\bf 3})-{\rm ch}_2(X)\right)\wedge{\rm ch}_2({\bf 8})+\frac{1}{5}{\rm ch}_4(X)\:.
\end{equation}
As ${\rm ch}_2({\bf 3})$ and ${\rm ch}_2(X)$ represent the same characteristic class, though constructed with different curvatures (${\rm ch}_2({\bf 3})$ is constructed using the curvature of $\tilde\nabla$ and ${\rm ch}_2(X)$ is constructed using the curvature of the background metric), the constraint \eqref{eq:topconstr1} is not satisfied in this case, and it might seem hard to cancel the anomaly using local and globally well-defined counter-terms. 

All is not lost however. Indeed, let $\alpha_0\in\Omega^{0,1}(\End(T^{1,0}X))$ be the dynamical part of the classical background. Consider the $(1,1)$-form
\begin{equation}
    \gamma=\tilde\nabla_a\alpha_0{}
    ^a{}_b\dd z^b\:,
\end{equation}
where $\tilde\nabla_a$ is the fixed background connection, and $\alpha_0$ is the dynamical part of this background, see footnote \ref{foot:Background}. A background gauge transformation of $\gamma$ then reads
\begin{equation}
    \delta\gamma=\tilde\nabla_a\delb\epsilon
    ^a{}_b\dd z^b=\tilde R_a{}^c{}_b\epsilon^a{}_c\dd z^b+\delb\left(\tilde\nabla_a\epsilon
    ^a{}_b\dd z^b\right)=\tr(\tilde R\epsilon)+\delb\left(\tilde\nabla_a\epsilon
    ^a{}_b\dd z^b\right)\:,
\end{equation}
where in the second equality we have used Proposition \ref{prop:Ricci}, that is $\tilde\nabla$ is "Ricci flat", $\tilde R_a{}^a{}_c=0$, and the last equality follows from Proposition \ref{prop:TorsionFree}, that is $\tilde\nabla$ is torsion-free. Given the anomaly polynomial \eqref{eq:anomalyPolSU3}, the descent procedure gives a gauge anomaly proportional to
\begin{equation}
    \int_X\left(3\,{\rm ch}_2({\bf 3})-{\rm ch}_2(X)\right)\wedge\tr(\tilde R\epsilon)\:,
\end{equation}
which may be cancelled by adding a counter-term proportional to
\begin{equation}
    \int_X\left(3\,{\rm ch}_2({\bf 3})-{\rm ch}_2(X)\right)\wedge\gamma\:.
\end{equation}

\subsection{Gravitational anomalies}
There is of course also the purely gravitational part of the anomaly polynomial to consider. For pure gauge theories, these are often ignored, or rather interpreted as a non-fatal ’t Hooft anomalies associated to the classical diffeomorphism symmetry of the classical background.\footnote{As noted above, the third, or "mixed anomaly", term of \eqref{eq:AnomalyPolHolCS} is often given the same treatment.} Still, it is desirable to also cancel the gravitational anomaly, especially as we will also consider the fully gravitational twisted theory of section \ref{sec:Convgauge} below. For holomorphic Chern--Simons theory, where the background geometry is represented by $\Omega$, we restrict to $\Omega$-preserving diffeomorphisms. The infinitesimal rotation of the holomorphic tangent bundle given by an infinitesimal diffeomorphism $v\in\Gamma(T)$ is then given by $\kappa^a{_b}=\nabla_a v^b$, where $\nabla_av^a=0$ ($\Omega$-preserving). I.e.  $\kappa$ sits in ${\rm End}_0(T^{1,0}X)$, or the Lie-algebra of $SU(3)$.

Canceling the gravitational anomaly for more general holomorphic Chern--Simons type theories seems tricky without the introduction of additional (gravitational) degrees of freedom and possibly resorting to Green--Schwarz type mechanisms such as in e.g. (holomorphic) heterotic or type I supergravity \cite{Costello:2015xsa, Costello:2016mgj, Costello:2019jsy, Ashmore:2023vji, Ashmore:2025fxr}. For example, if we introduce a number of bosonic degrees of freedom $\kappa_i\in\Omega^0(X)$ and $\beta_i\in\Omega^{0,2}(X)$ to standard holomorphic Chern--Simons theory, where $i\in\{1,2,\dots,N\}$ with kinetic terms
\begin{equation}
    S_{\rm kin}(\kappa_i,\beta_i)=\int_X\kappa_i\delb\beta_i\wedge\Omega\:,
\end{equation}
then the only effect this has on the anomaly polynomial is to correct the gravitational part to\footnote{The minus in the denominator is due to the fact that the classical bosonic degrees of freedom $\kappa_i$ and $\beta_i$ are differential forms shifted by one degree.}
\begin{equation}
    \frac{{\rm dim}(\Lg)}{40}{\rm ch}_4(X)\rightarrow\frac{{\rm dim}(\Lg)-2N}{40}{\rm ch}_4(X)\:.
\end{equation}
For even-dimensional Lie-algebras this anomaly may then be canceled simply by setting $2N={\rm dim}(\Lg)$.

Of course, coupling the theory to additional degrees of freedom after we have deformed, is a bit unsatisfactory. Instead, what we should do is couple the original undeformed theory to a set of BV theories of the form
\begin{equation}
    S_{\rm kin}(\kappa_i,\beta_i)=\int_X\kappa_i\dd\beta_i\wedge\Omega\:,
\end{equation}
where now $\kappa_i\in\Omega^0(X)$ and $\beta_i\in\Omega^{2}(X)$ by slight abuse of notation. After deforming, and after a somewhat lengthy but straight forward computation involving field redefinitions ala the start of section \ref{sec:Quantum}, which we will skip, the action becomes
\begin{equation}
    S_{\rm kin}(\kappa_i,\beta_i)=\vh\int_X\kappa_i D\beta^{2,0}_i\wedge\bar\Omega\:,
\end{equation}
where $D=\del+h^{\bar a}\delb_{\bar a}$. The dominant contribution for large $\vh$ is then simply
\begin{equation}
    S_{\rm kin}(\kappa_i,\beta_i)=\vh\int_X\kappa_i\del\beta^{2,0}_i\wedge\bar\Omega\:.
\end{equation}
If we then BV-quantize as in section \ref{sec:BV-Metric}, the partition function \eqref{eq:PartFuncBVMetric} is corrected to
\begin{equation}
\label{eq:PartFuncWextra1}
    Z^{\rm BV}_2=\vh^{\frac{2N-{\rm dim}(\Lg)}{480}\chi(X)}\int{\cal D}\Phi\,{\rm exp}\left(i\int_X\Big(\tr(\gamma\del_A^{\rm BV}\gamma)+\sum_i\kappa_i\del^{\rm BV}\beta_i\Big)\wedge\bar\Omega\right)\:,
\end{equation}
where now $\gamma$ is a poly-form in $\Omega^{\bullet,0}(\End(E))$, and $\kappa_i$ and $\beta_i$ are poly-forms in $\Omega^{\bullet,0}(X)$. Furthermore, $\gamma$ is bosonic in odd degree, and fermionic in even degree, and opposite for $\kappa_i$ and $\beta_i$. Here $\Phi$ is short hand notation for all the fields involved, and the path integral is taken over the kernel of the corresponding adjoint operators, $\del^\dagger_A$ for $\gamma$ and $\del^\dagger$ for $\kappa_i$ and $\beta_i$. 

Canceling the gravitational anomaly in this way then gives the partition function
\begin{equation}
    Z^{\rm BV}_2=\int{\cal D}\Phi\,{\rm exp}\left(i\int_X\Big(\tr(\gamma\del_A^{\rm BV}\gamma)+\sum_i\kappa_i\del^{\rm BV}\beta_i\Big)\wedge\bar\Omega\right)\:,
\end{equation}
which also has the effect of removing the explicit dependence of the partition function on $\vh$.

\subsubsection*{Anomalies in the "twisted theory"}
Before we move to discuss geometric anomalies, let's discuss the $\End(T^{1,0}X)$-theory in the convenient gauge of section \ref{sec:Convgauge} (where again we pick $h$ in the special directions making $\tilde\nabla$ pseudo-hermitian). In this gauge, the partition function becomes
\begin{equation}
    Z^{\rm BV}_2=\vh^{\frac{2N-{\rm dim}({\bf 8})}{480}\chi(X)}\int{\cal D}\Phi\,{\rm exp}\left(i\int_X\Big(\tr(\gamma\del_\nabla^{\rm BV}\gamma)+\sum_i\kappa_i\del^{\rm BV}\beta_i\Big)\wedge\bar\Omega\right)\:,
\end{equation}
Here again $\del_\nabla$ is the Chern connection of the background metric, and the integral over $\gamma$ is now restricted to the kernel of $\del_\nabla^{\dagger}$. The full anomaly polynomial of the corresponding holomorphic Chern--Simons theory is now purely gravitational, reading
\begin{equation}
    {\cal P}=\frac{248-2N}{480}{\rm ch}_2(X)^2\:.
\end{equation}
We see than that for this "twisted theory", picking $2N=248$ cancels the anomaly in full. With this choice, the partition function becomes
\begin{equation}
\label{eq:PartFuncWextra2}
    Z^{\rm BV}_2=\vh^{\tfrac12\chi(X)}\int{\cal D}\Phi\,{\rm exp}\left(i\int_X\Big(\tr(\gamma\del_\nabla^{\rm BV}\gamma)+\sum_{i=1}^{124}\kappa_i\del^{\rm BV}\beta_i\Big)\wedge\bar\Omega\right)\:.
\end{equation}
The observant reader will perhaps pounder as to the significance of the number $248$. Indeed, one can imagine collecting the $\kappa_i$ and $\beta_i$ fields in a $248$-dimensional multiplet, which then enjoys a global $E_8$-symmetry.\footnote{To get a term as in \eqref{eq:PartFuncWextra2}, one needs to use a skew-symmetric pairing on the $E_8$ Lie-algebra.} This symmetry can then be gauged using a flat $E_8$-connection. As the bundle is flat with vanishing curvature, this will not alter the form of the anomaly. It is interesting that holomorphic Chern--Simons theory on ${\rm End}_0(T^{1,0}X)$ can be made anomaly free in this way, by coupling to a fermionic-type $E_8$-theory,\footnote{Here "fermionic-type", refers to that classical fields of this theory are shifted by one degree with respect to ordinary holomorphic Chern--Simons.} though we do not expect to be the first to make this observation.

\subsection{Geometric anomalies}
\label{sec:GeomAnomalies}
We now come to consider geometric anomalies. These are not anomalies in the usual sense, but rather a dependence of the (absolute value of the) partition function of the background geometric structures. In particular, in quantizing the theory we have to pick metrics on both the manifold and the bundle in order to gauge-fix. Furthermore, as noted above, the generic partition function, for example \eqref{eq:PartFuncWextra1}, will have an implicit dependence of $h$ due to the background connection $\del_A$ appearing in the action. Again, this connection solves the instanton equation of motion \eqref{eq:Instanton}, making the background connection $h$-dependent. Moreover, $h$ also depends on the background metric (though $\vh$ does not), as $h$ is chosen harmonic with respect to the background metric, complicating things even further. Finally, the partition function also depends on the background complex structure, as worked out by Bismut etal \cite{bismut1988analytic1, bismut1988analytic2, bismut1988analytic3}. 

All these considerations make a general analysis of the geometric dependence of the partition function of the generic theory rather daunting. Here, we will instead focus on the ${\rm End}_0(T^{1,0}X)$-theory, which is the only explicit example of a "deformed Chern--Simons theory" we have presented, and leave the general analysis to future work. Working with the example of the "twisted theory" \eqref{eq:PartFuncWextra2}, where the implicit $h$-dependence drops out, will also simplify matters. Furthermore, this theory has a vanishing anomaly polynomial, which turns out also to simplify both the metric and background complex structure dependence of the partition function. Though, as noted above, we are restricted to picking one of the directions $h$ of the deformation parameter where $\tilde\nabla$ becomes self-adjoint.

\subsubsection*{Metric anomaly}
Let's first consider the dependence of the partition \eqref{eq:PartFuncWextra2} of the background metric, or "metric anomaly". Including the volumes of harmonic forms, the absolute value of the partition function reads
\begin{equation}
    \vert Z^{\rm BV}_2\vert^2=\vvh^{\chi(X)}\times\bigg(\frac{{\rm Vol}_{g^\bullet}\left(Z^\bullet({\rm End}_0(T^{1,0}X))\right)}{\left({\rm Vol}_{g^\bullet}(Z^\bullet(X))\right)^{248}}\bigg)^{\tfrac12}\times\frac{{\cal I}({\rm End}_0(T^{1,0}X))}{{\cal I}^{248}_0}\:.
\end{equation}
Note that the volume factors in the twisted theory should also be worked out using the ordinary metric $g$ on ${\rm End}_0(T^{1,0}X)$. It is clear that $\vh$, being proportional to the Yukawa coupling of $h$, see equation \eqref{eq:vh-Yuk}, is metric independent. However, the volume factors and holomorphic Ray--Singer torsions do depend on the metric $g$ on $X$, where the explicit formulas where worked out by Bismut etal \cite{bismut1988analytic1, bismut1988analytic2, bismut1988analytic3}.

For general deformations of the metric $g$ on $X$ (not just rescalings), the metric anomaly formula for the holomorphic Ray--Singer torsion ${\cal I}(E)$ are best presented by considering the Quillen measure
\begin{equation}
    {\cal Q}(E)={\rm det}(g^{\bullet}(E))^{\tfrac12}\,{\cal I}(E)\:,
\end{equation}
where ${\rm det}(g^{\bullet}(E))$ is short-hand notation for
\begin{equation}
    {\rm det}(g^{\bullet}(E))=\frac{{\rm det}(g^1(E)){\rm det}(g^3(E))}{{\rm det}(g^0(E)){\rm det}(g^2(E))}\:,
\end{equation}
where ${\rm det}(g^p)$ is the determinant of the metric viewed as a hermitian inner product on ${\cal H}^{p,0}(E)$, and where the metric $g^p(E)$ on harmonic forms is given by an analogous formula to \eqref{eq:IPh}. The metric anomaly formula then reads \cite[Theorem 1.22]{bismut1988analytic1, bismut1988analytic2, bismut1988analytic3}
\begin{equation}
\label{eq:VarTorsion}
    \frac{1}{\pi} \delta_g \log {\cal Q}(E) = \int_X \partial_t \left[ {\rm td} \left(\frac{1}{2\pi} (iR+t\,g^{-1}\delta g)\right)\wedge {\rm ch}\left(\frac{1}{2\pi} (iF+t\, {\cal H}^{-1}\delta{\cal H})\right) \right]_{4}\bigg\vert_{t=0}  \, ,
\end{equation}
where $R$ and $F$ are the curvatures of the Chern connections on the holomorphic tangent bundle and $E$ respectively, $\cal H$ is a hermitian metric on $E$, and the total Todd class and Chern character are expanded in terms of the indicated arguments. The subscript $4$ means that we take the component that is degree 4 in the curvatures $R,F$ and their variations. For the theory \eqref{eq:PartFuncWextra2}, $\cal H$ is also given in terms of the metric $g$ on ${\rm End}_0(T^{1,0}X)$. 

To apply these formulas for general variations, it is convenient to promote the one--loop partition function to the Quillen measure instead, without the integration over harmonic forms, see for example \cite{Bershadsky:1993cx, muller1999extremal}. For the theory at hand, we then get
\begin{equation}
\label{eq:QuillenPF}
    \vert Z^{\rm BV}_2\vert^2=\vvh^{\chi(X)}\times\frac{{\cal Q}({\rm End}_0(T^{1,0}X))}{{\cal Q}_0^{248}}\:,
\end{equation}
where ${\cal Q}_0$ is the Quillen measure for forms not valued in a bundle. The observant reader will note that the variation \eqref{eq:VarTorsion} essentially corresponds to a variation of the anomaly polynomial of the corresponding holomorphic Chern--Simons theory on $E$. One might therefore think that the variation vanishes for the full theory \eqref{eq:PartFuncWextra2}. This is almost correct, with the caveat that $g^{-1}\delta g$ will in general also include a non-trivial trace, or "singlet" under the $SU(3)$ decomposition. The singlet is proportional to the trace $\tr(g^{-1}\delta g)$, where we trace over the fundamental of $SU(3)$. For Kähler Ricci-flat deformations, this is a constant on $X$. Moreover, given the volume
\begin{equation}
    {\rm Vol}_g(X)=\tfrac{1}{3!}\int_X\omega\wedge\omega\wedge\omega\:,
\end{equation}
we find that
\begin{equation}
\label{eq:VarVol}
    \delta\log\left({\rm Vol}_g(X)\right)=\tr(g^{-1}\delta g)\:.
\end{equation}
A metric variation may then be decomposed as
\begin{equation}
    g^{-1}\delta g=\tfrac13\tr(g^{-1}\delta g)\,{\bf 1}+\left(g^{-1}\delta g\right)_{\bf 8}\:,
\end{equation}
where the last term denotes the adjoint part of the variation. This term essentially acts as a gauge transformation, and the corresponding metric anomaly associated to this term therefore vanishes for the theory \eqref{eq:QuillenPF}.

Using \eqref{eq:VarTorsion}, we then find the metric variation for the partition function to be 
\begin{equation}
    \delta_g\log\left(\vert Z^{\rm BV}_2\vert^2\right)=-\tfrac16\,\chi(X)\,\tr(g^{-1}\delta g)\:,
\end{equation}
which, using \eqref{eq:VarVol}, implies that the combination
\begin{equation}
\label{eq:VarQX}
    \vert Z^{\rm BV}_2\vert^2\,{\rm Vol}_g(X)^{\tfrac16\chi(X)}
\end{equation}
is independent of (Kähler) metric deformations. This is similar to the volume factor used to make the one-loop partiton function of the type IIB Hithcin functional metric independent \cite{Pestun:2005rp}.

\subsubsection*{Holomorphic anomaly}
We can also consider the holomorphic anomaly. That is, how does the partition function depend on the background complex structure. The holomorphic anomaly formula for the Quillen measure for a holomorphic bundle $E$ now reads \cite{bismut1988analytic1, bismut1988analytic2, bismut1988analytic3}
\begin{equation}
    \frac{1}{2\pi i}\bdel\bdelb\log{\cal Q}(E)=\frac{1}{2}\int_X\left[ {\rm td} \left(\frac{iR}{2\pi}\right)\wedge {\rm ch}\left(\frac{iF}{2\pi}\right) \right]_{4}\:,
\end{equation}
where we should view the curvatures $R$ and $F$ as curvatures of the holomorphic tangent bundle $T^{1,0}X$ and gauge bundle $E$, extended to the universal geometry, where we think of $X$ as fibered over the moduli space. The derivatives $\bdel$ and $\bdelb$ are now in moduli directions,
\begin{equation}
    \bdel=\dd Z^A\frac{\del}{\del Z^A}\:,\;\;\;\;\bdelb=\dd {\bar Z}^{\bar A}\frac{\del}{\del{\bar Z}^{\bar A}}\:.
\end{equation}
Applied to the case at hand, we then need the curvature $\boldsymbol{\cal R}$ of the holomorphic tangent bundle $T^{1,0}X$ extended to the moduli space. This can be decomposed as
\begin{equation}
    \boldsymbol{\cal R}={\boldsymbol{\cal R}}_{\bf 8}+\tfrac13\tr({\boldsymbol{\cal R}})\,{\bf 1}\:.
\end{equation}
As for metric deformations, in the combination of Quillen measures in \eqref{eq:QuillenPF} the adjoint part of the curvature drops out. Moreover, we have \cite{Bershadsky:1993cx}
\begin{equation}
    \frac{i}{2\pi}\tr({\boldsymbol{\cal R}})=c_1(T^{1,0}X)=-c_1(T^{*1,0}X)=-\frac{\boldsymbol{\cal\omega}}{2\pi}\:,
\end{equation}
where ${\boldsymbol{\cal\omega}}=i\bdel\bdelb K$ is the Kähler form on the moduli space. Combining this, we find
\begin{equation}
\label{eq:VarQZ}
    i\bdel\bdelb\log\left(\frac{{\cal Q}({\rm End}_0(T^{1,0}X))}{{\cal Q}_0^{248}}\right)=-\tfrac16\chi(X){\boldsymbol{\cal\omega}}\:.
\end{equation}
Combining \eqref{eq:VarQX} and \eqref{eq:VarQZ}, we see that the moduli dependence of the combination of Quillen measures goes like
\begin{equation}
\label{eq:QModDep}
    \log\left(\frac{{\cal Q}({\rm End}_0(T^{1,0}X))}{{\cal Q}_0^{248}}\right)\sim \tfrac16\chi(X)\left(K(X^A)-K(Z^A)\right)\:,
\end{equation}
where $X^A$ are Kähler moduli, with corresponding Kähler potential
\begin{equation}
    e^{-K(X^A)}={\rm Vol}_g(X)=\tfrac{1}{3!}\int_X\omega\wedge\omega\wedge\omega\:.
\end{equation}

This is perhaps expected from mirror symmetry. Indeed, mirror symmetry switches Kähler and complex structure moduli, and flips the sign of the Euler number. Equation \eqref{eq:QModDep} then expresses the, perhaps naiive, expectation that the topological open string one-loop partition function is invariant under mirror symmetry, where the path integral of \eqref{eq:PartFuncWextra2} denotes the corresponding target space theory of the topological open string theory in question.

As an aside, one can also consider the type II topological string, where the one-loop free energy of the type IIB depends on complex structure moduli as \cite{Bershadsky:1993cx}
\begin{equation}
    F_{\rm IIB}(X, Z^A)\sim \tfrac{1}{24}\chi(X)K(Z^A)\:.
\end{equation}
However, ignoring subtleties related to harmonic zero-modes, one finds that the type IIA and type IIB one-loop partition functions agree at large volume on the same Calabi--Yau X \cite{Bershadsky:1993cx, deBoer:2007zu, Ashmore:2021pdm}. The large volume Kähler moduli dependence of the type IIB one-loop free energy is therefore
\begin{equation}
    F_{\rm IIB}(X,X^A)=F_{\rm IIA}(\tilde X,\tilde Z^A)=F_{\rm IIB}(\tilde X, \tilde Z^A)\sim \tfrac{1}{24}\chi(\tilde X)K(\tilde Z^A)=-\tfrac{1}{24}\chi(X)K(X^A)\:,
\end{equation}
where the first equality is mirror symmetry, and the minus in the last equality is due to $\chi(\tilde X)=-\chi(X)$. Here $\tilde X$ denotes the mirror of $X$, together with mirror Kähler moduli $\tilde X^A$ and complex structure moduli $\tilde Z^A$. Hence,
\begin{equation}
    F_{\rm IIB}(X, Z^A, X^A)\sim\tfrac{1}{24}\chi(X)\left(K(Z^A)-K(X^A)\right)\:,
\end{equation}
which is also invariant under mirror symmetry with the observation that the type IIB and type IIA one-loop partition functions agree on the same Calabi--Yau $X$ at large volume. 

Returning to the case at hand, there is also the $\vvh$-dependent prefactor of the partition functction \eqref{eq:QuillenPF} to consider. Note that $\vh$, given by equation \eqref{eq:vh-Yuk}, depends both on background moduli through the Yukawa couplings and Kähler potential, and on the special directions in deformation space where the corresponding connection on ${\rm End}_0(TX)$ is self-adjojnt. How these special directions change as we move around in complex structure moduli space, and if they may be extended to global sections of ${\cal T}{\cal M}$ will be the subject of future publications.

\section{Conclusions and outlook}
In this work we have investigated a deformation of holomorphic Chern–Simons theory obtained by explicitly deforming the background complex structure of a Calabi–Yau threefold. Rather than treating complex structure dependence indirectly through variation of parameters, we constructed a fully deformed classical action by expanding the holomorphic top-form in powers of a complex structure deformation tensor and inserting this expansion directly into the Chern–Simons functional. Although this theory can be viewed formally as ordinary holomorphic Chern–Simons theory written in terms of a deformed complex structure, the explicit parametrization in terms of the deformation tensor $h^a_{\bar b}$ reveals structural features that are far from obvious in the conventional formulation. In particular, the deformed action depends simultaneously on both holomorphic and anti-holomorphic components of the gauge connection, and the resulting equations of motion exhibit a rich interplay between curvature components of different Hodge types.

At the classical level, we derived a compact and equivalent formulation of the equations of motion, which makes manifest the role of the deformation parameter $h\in{\cal H}^{0,1}(T^{1,0}X)$ (Beltrami differential) and its associated $(2,1)$-form $\chi=\Omega(h)$. Among the solutions, we identified a distinguished class of scale-invariant "instanton" solutions
\begin{equation}
    F^{0,2} = 0\:,\;\;\; F\wedge\chi = 0\:,
\end{equation}
characterized by the holomorphic bundle condition, where the last equation replaces the usual primitive condition for Yang--Mills instantons. These instantons are invariant under rescalings of the deformation parameter and therefore depend only on its direction in complex structure moduli space. Structurally, the resulting equations bear resemblance to higher-dimensional gauge-theoretic instanton conditions, particularly those encountered in exceptional $G_2$ holonomy setting, suggesting that the deformation tensor could play a role analogous to a calibration form selecting preferred curvature components.

We demonstrated that non-trivial solutions to these instanton equations do exist. Besides simple abelian examples, we constructed non-abelian solutions on $\End(T^{1,0}X)$. A central observation is that when the deformation parameter has non-vanishing Yukawa coupling, it induces a natural Chern-type connection on the tangent bundle whose curvature satisfies the instanton constraint. However, compatibility with a real self-adjoint gauge theory structure on the underlying real bundle $\End(TX)$ is not automatic. Instead, such self-adjointness holds only for special directions in complex structure deformation space, determined by a non-trivial condition relating the deformation parameter to its complex conjugate. This condition singles out specific rays in the projectivized deformation space, and it is precisely along these directions that the induced connection becomes a self-adjoint (though pseudo-hermitian) $SU(3)$-connection on the trace-free endomorphism bundle.

The analysis of these special directions reveals an intriguing connection between gauge theory and the global geometry of complex structure moduli space. By reformulating the hermiticity constraint as a critical point condition for a homogeneous functional constructed from Yukawa couplings and the moduli space metric, we were able to interpret the problem variationally. The resulting function descends to projective moduli space and, away from the vanishing locus of the Yukawa coupling, behaves as a Morse function. This perspective provides both existence results and lower bounds on the number of special directions via relative Morse theory. We gave an explicit low-dimensional example of this. This suggests that the space of physically distinguished gauge backgrounds is controlled by topological features of the complex structure moduli space.

Turning to the quantum theory, we quantized the deformed action around such scale-invariant instanton backgrounds using the background field method. After a sequence of field redefinitions, the kinetic and interaction terms simplify dramatically and can be written in a form closely resembling anti-holomorphic Chern–Simons theory, multiplied by the determinant of the deformation tensor $h_{\bar b}^a$. The nilpotency of the corresponding kinetic operator relies crucially on the instanton condition, illustrating the connection between the classical background equations and the structure of the quantum fluctuations. We carried out the quantization using complementary approaches, including formal one-loop analysis and BRST and BV methods. The resulting one-loop partition function can be expressed in terms of generalized holomorphic Ray–Singer torsions associated with the deformation-dependent differential operator. In the limit of large complex structure deformation, the expressions simplify considerably and reduce to torsions of standard Dolbeault-type operators multiplied by explicit powers of the deformation determinant. This large-deformation regime thus provides a controlled approximation in which the dependence on the deformation parameter becomes particularly transparent. Nevertheless, even in this limit, an implicit dependence remains through the background connection, which itself satisfies the deformation-dependent instanton condition. However, in the case of the above mentioned theory on $\End(T^{1,0}X)$, via a suitable field redefinition, we found that this implicit dependence drops out for the special directions in the deformation space.

A central theme of the quantum analysis is the study of anomalies and geometric dependence. Because the deformation modifies both kinetic operators and the functional measure, potential gauge, gravitational, and geometric anomalies must be reexamined. The analysis simplifies for large complex structure deformations, and the usual anomaly analysis for real self-adjoint gauge theories can be applied. In particular, our analysis shows that along the special directions in deformation space, the theory defined on $\End(T^{1,0}X)$ can be made anomaly free by coupling the theory to an additional $E_8$-worth of gravitational degrees of freedom. We also studied geometric anomalies and the metric dependence of the resulting partition function, showing a mild dependence proportional to the Euler number as in the type B closed topological string \cite{Bershadsky:1993cx, Pestun:2005rp}.

Several broader perspectives emerge from this work. First, the deformation tensor $h^a_{\bar b}$, usually regarded merely as a parameter describing variations of complex structure, acquires an active dynamical role when inserted explicitly into the gauge theory action. Through its Yukawa coupling and its interaction with curvature components, it governs both the classical instanton structure and the quantum kinetic operator. This suggests perhaps that complex structure deformations and gauge dynamics are more tightly interwoven than is apparent in the standard formulation of holomorphic Chern–Simons theory. In particular, the appearance of a Morse-theoretic structure on projective deformation space hints at a deeper variational principle relating special geometry data to gauge-theoretic stability conditions and instanton configurations.

Furthermore, the emergence of instanton equations resembling those in higher-dimensional gauge theories raises the possibility that deformed holomorphic Chern–Simons theory may provide a bridge between complex three-dimensional gauge theory and exceptional holonomy constructions. The contraction of the curvature with the deformation tensor plays a role analogous to projection onto calibrated components, and it would be interesting to understand whether this analogy can be sharpened, perhaps in the context of dimensional reductions or dualities.

On the classical side, a more detailed study of the moduli space of instanton solutions would clarify questions of uniqueness, rigidity, and obstruction theory. Is there a notion of stability for such instantons? The conjecture that the special directions exhaust all solutions compatible with real gauge structure on $\End(TX)$ deserves a more complete proof. It would also be valuable to analyze the behavior of the Morse function governing special directions as one moves over the full complex structure moduli space, and to determine whether global sections of the corresponding critical loci exist or whether monodromies and wall-crossing phenomena occur. It is also worth investigating if the special directions have connections to the attractor mechanism and attractor flow, given the similar nature of the corresponding equations.  

On the quantum side, the interplay between the deformation determinant, Yukawa couplings, and analytic torsion suggests potential connections with holomorphic anomaly equations and topological string amplitudes. In this regard, it would be interesting to consider the full cubic theory and higher loop anomalies, away from the large complex structure limit of the deformation parameter. Since the topological open string underlies holomorphic Chern–Simons theory, a natural next step to investigate is how the explicit deformation considered here encodes novel aspects of complex structure dependence at the level of open topological string theory. Furthermore, it is clear that the special directions have interesting properties from the point of view of complex structure moduli. It would be interesting to investigate what the corresponding properties are, if any, for Kähler moduli on the mirror dual side. 

Finally, it would be interesting to explore whether similar constructions can be extended beyond Calabi–Yau three-folds, for instance to higher-dimensional Calabi–Yau manifolds, exceptional holonomy, or to settings with additional fluxes or torsion. E.g. can the special self-adjoint connections on $\End(TX)$ be used in heterotic string theory, perhaps analogusly to the Standard Embedding \cite{Candelas:1985en}? The structural simplicity of the deformed action after suitable field redefinitions suggests that the essential mechanism may be robust and adaptable. In this broader context, deformations of background geometric structures could provide a systematic way of generating new gauge-theoretic models whose classical and quantum properties are tightly constrained by special geometry.

\section*{Acknowledgements}
We thank Anthony Ashmore, Alex S. Arvanitakis, Philip Candelas, Xenia de la Ossa,  Henrique N. Sá Earp, Kristoffer Rørhus Flaate, Mario Garcia-Fernandez, Pedram Hekmati, Johanna Knapp, Jason D. Lotay, Jock McOrist, Javier José Murgas Ibarra, Paul-Konstantin Oehlmann, Fabian Ruehle,  Maik Sarve, and Charles Strickland-Constable for useful and inspiering discussions surrounding the work presented here.

\appendix

\section{Some special geometry relations}
\label{app:spacial}
In this appendix, we review some useful relations for the special geometry of the complex structure moduli space, that we have used throughout this paper. The reader is referred to \cite{Candelas:1990pi, Strominger:1990pd} for more complete introductions. 

We begin by recalling the definition of the Kähler potential
\begin{equation}
\label{eq:KahlerPotential}
    e^{-K}=i\int_X\Omega\wedge\bar\Omega\:.
\end{equation}
A variation of the $(3,0)$-form may then be written as
\begin{equation}
    \del_A\Omega=K_A\Omega+\chi_A\:,
\end{equation}
where $K_A$ corresponds to the $(3,0)$-part of the variation, and 
\begin{equation}
    \chi_A=\tfrac12\,\Omega_{abc}h_A^a\,\dd z^{bc}\:,
\end{equation}
where 
\begin{equation}
    h_A^a=\del_A(\dd z^a)
\end{equation}
is a $T^{1,0}X$-valued $(0,1)$-form corresponding to a deformation of a complex structure. In harmonic gauge, which we are working in throughout this paper, $K_A$ is constant on the manifold $X$, while $\chi_A$ and $h_A$ are harmonic.

Using the expression for the Kähler potential \eqref{eq:KahlerPotential}, one finds that
\begin{equation}
    \del_AK=-ie^K\int_X(\del_A\Omega)\wedge\bar\Omega=-K_A\:.
\end{equation}
It follows that
\begin{equation}
    \chi_A={\cal D}_A\Omega=\del_A\Omega+\del_AK\,\Omega\:,
\end{equation}
where one view $\Omega$ as a section of a holomorphic line bundle $\cal L$ over the moduli space of complex structures. Here $\del_AK$ gives the connection one-form of the covariant derivative $\cal D$ on this line bundle. 

The Kähler potential gives rise to a Kähler metric on the complex structure moduli space, given by
\begin{equation}
    G_{A\bar B}=\del_A\del_{\bar B}K=-ie^K\int_X\chi_A\wedge\bar\chi_{\bar B}\:.
\end{equation}
This metric is used to extend $\cal D$ to other tensors on the complex structure moduli space. For example
\begin{equation}
    {\cal D}_A\chi_B=\del_A\chi_B+\del_A K\,\chi_B-\Gamma_A{}^C{}_B\,\chi_C\:,
\end{equation}
where $\Gamma_A{}^C{}_B$ are the connection symbols of the Levi-Civita connection of $G_{A\bar B}$.

We then recall some useful relations in special geometry
\begin{align}
    {\cal D}_A\Omega&=\chi_A\\
    {\cal D}_B{\cal D}_A\Omega&=-ie^K\kappa_{ABC}G^{C\bar D}\bar\chi_{\bar D}
    \label{eq:Cov2}\\
    {\cal D}_{A}\bar\chi_{\bar B}=\del_{A}\bar\chi_{\bar B}&=G_{A\bar B}\,\bar\Omega\:,
    \label{eq:Cov3}
\end{align}
where 
\begin{equation}
\kappa_{ABC}=\int_X\del_A\del_B\del_C\Omega\wedge\Omega=\int_X\Omega(h_A,h_B,h_C)\wedge\Omega=\int_X\tfrac{1}{3!}h_A^ah_B^bh_C^c\Omega_{abc}\wedge\Omega
\end{equation}
are the Yukawa couplings. 

We can use the above relations to derive some useful results, used throughout this paper. Let's first expand $h^a$ in the basis of harmonic forms $\{h_A^a\}$ as
\begin{equation}
    h^a=Z^A\,h_A^a\:.
\end{equation}
We then have the following
\begin{proposition}
\label{prop:h-const}
    Let $\vert h\vert$ be given by
    \begin{equation}
        \Omega(h,h,h)=\tfrac{1}{3!}h^ah^bh^c\Omega_{abc}=\vert h\vert\bar\Omega\:.
    \end{equation}
    Then $\vert h\vert$ is a constant on the manifold $X$.
\end{proposition}
\begin{proof}
    From equation \eqref{eq:Cov2} we get
    \begin{align}
        {\cal D}_C{\cal D}_B{\cal D}_A\Omega&={\cal D}_C\left(-ie^K\kappa_{ABD}G^{D\bar D}\right)\bar\chi_{\bar D}-ie^K\kappa_{ABD}G^{D\bar D}\del_C\bar\chi_{\bar D}\notag\\
        &={\cal D}_C\left(-ie^K\kappa_{ABD}G^{D\bar D}\right)\bar\chi_{\bar D}-ie^K\kappa_{ABC}\,\bar\Omega\:,
    \end{align}
    where we have used equation \eqref{eq:Cov3}. Taking the $(0,3)$-parts on both sides, we see that
    \begin{equation}
        \Omega(h_A,h_B,h_C)=\tfrac{1}{3!}h_A^ah_B^bh_C^c\Omega_{abc}=-ie^K\kappa_{ABC}\bar\Omega.
    \end{equation}
    It follows that
    \begin{equation}
        \vert h\vert\bar\Omega=\Omega(h,h,h)=-ie^K\kappa\,\bar\Omega\:,
    \end{equation}
    where $\kappa=\kappa_{ABC}Z^AZ^BZ^C$. That is
    \begin{equation}
        \vert h\vert=-ie^K\kappa\:,
    \end{equation}
    which is constant on $X$.
\end{proof}
From the above proof, it is also clear that $\Omega(h_A,h_B,h_C)$ are harmonic $(0,3)$-forms, and therefore proportial to $\bar\Omega$, with proportionality constant $-ie^K\kappa_{ABC}$.

Next we have the following proposition\footnote{We thank Xenia de la Ossa for making us aware of this useful relation.} 
\begin{proposition}
\label{prop:GAB}
    The Käher metric $G_{A\bar B}$ can be expressed in terms of $h_A^a$ as
    \begin{equation}
        G_{A\bar B}=h^a_{A\,\bar b}\bar h^{\bar b}_{\bar B\,a}\:,
    \end{equation}
    which is then a constant on $X$.    
\end{proposition}
\begin{proof}
    Consider equation \eqref{eq:Cov3}, which can be written as
    \begin{equation}
    \label{eq:Cov3-2}
        \del_{\bar B}\chi_A=G_{A\bar B}\,\Omega\:.
    \end{equation}
    The $(3,0)$-part of the left hand side is then given by
    \begin{equation}
        \tfrac12\Omega_{abc}h_{A\,\bar b}^a\,\del_{\bar B}(d \bar z^{\bar b})\dd z^{bc}=\tfrac12\Omega_{abc}h_{A\,\bar b}^a\bar h^{\bar b}_{\bar B\,d}\dd z^{dbc}=h^a_{A\,\bar b}\bar h^{\bar b}_{\bar B\,a}\,\Omega\:,
    \end{equation}
    where the last equality follows from some elementary index manipulation. Comparing with equation \eqref{eq:Cov3-2} gives the result. 
\end{proof}

\section{Geometric properties of $h_{\bar a}^b$}
\label{app:Geometry}
In this appendix, using special geometry, we note some useful geometric features related to the deformation parameter $h$, and its associated geometric properties. 

We begin by noting some useful facts about the deformation parameter $h$, which we shall take to be {\it harmonic} with respect to the background Calabi-Yau metric $g$. Firstly,
\begin{equation}
    h_{\bar a\bar b}=g_{\bar a c}{h_{\bar b}}^c
\end{equation}
is {\it symmetric} in its indices. Indeed the anti-symmetric part are given by harmonic $(0,2)$-forms, which vanish on a Calabi-Yau of full $SU(3)$ holonomy. 

As stated in the text, we take ${h_{\bar a}^b}$ to be an invertible matrix, with "determinant" given by $\vh\neq0$, where
\begin{equation}
    \Omega(h,h,h)=\tfrac{1}{3!}h^ah^bh^c\Omega_{abc}=\vh\,\bar\Omega\:.
\end{equation}
Special geometry dictates that $\vh$ is a constant on the manifold $X$, as shown in Proposition \ref{prop:h-const} above. The inverse of ${h_{\bar a}^b}$ is ${h_{c}^{\bar d}}$, where
\begin{equation}
    {h_{c}^{\bar a}}{h_{\bar a}^b}=\delta_c^b\:.
\end{equation}
Let's expand $h$ in a background basis $\{h_A\}$ of harmonic forms ${\cal H}^{0,1}(T^{1,0}X)$. That is
\begin{equation}
    h^a=Z^B\,h^a_{B}\:.
\end{equation}
We have the following
\begin{proposition} \label{prop:inverse}
The inverse of the matrix ${h_{\bar a}^b}$ is given by the following formula
    \begin{equation}
\label{eq:h-Inv}
        {h_{b}^{\bar c}}=\frac{3}{\kappa}{\bar h}_{\bar B\,b}^{\bar c}\,G^{\bar B C}\kappa_{C}\:.
    \end{equation}
Here  $G_{\bar A B}$ is the metric on the Calabi-Yau moduli space, while $\kappa$ is the Yukawa coupling
    \begin{equation}
    \kappa={\rm Yuk}(h,h,h)=Z^AZ^BZ^C\kappa_{ABC}\:.
    \end{equation}
We also define $\kappa_A=\kappa_{ABC}Z^bZ^C$ and $\kappa_{AB}=\kappa_{ABC}Z^C$.
\end{proposition}
\begin{proof}
Using equation \eqref{eq:Cov2} above, we get
\begin{equation}
    \tfrac12h_{\bar a}^ah^b_{\bar b}\,\Omega_{abc}=-ie^K\kappa_CG^{C\bar D}\tfrac12{\bar h}_{\bar D\,c}^{\bar e}\bar\Omega_{\bar e\bar a\bar b}\:.
\end{equation}
Note that we must also have
\begin{equation}
    h_{c}^{\bar d} \propto \Omega^{\bar d\bar a\bar b}h_ {\bar a}^ah_{\bar b}^b\Omega_{abc}\:,
\end{equation}
since
\begin{equation}
    \Omega^{\bar d\bar a\bar b}h_ {\bar a}^ah_{\bar b}^b\Omega_{abc}h_{\bar e}^c\propto\vert h\vert\Omega^{\bar d\bar a\bar b}\bar\Omega_{\bar a\bar b\bar e}\propto\delta_{\bar e}^{\bar d}\:.
\end{equation}
It follows that
\begin{equation}
    h_{a}^{\bar b}\propto-ie^K\bar h_{\bar B\,a}^{\bar b}{\kappa^{\bar B}}\:.
\end{equation}
To find the proportionality constant, we use the relations
\begin{equation}
    h_{\bar b}^ah_a^{\bar b}=3\:,\;\;\;\bar h_{\bar B\,a}^{\bar b}h_{D\,\bar b}^a=G_{\bar BD}\:,
\end{equation}
where the last relation again follows from special geometry, see Proposition \ref{prop:GAB}.
\end{proof}

Note that we can also write the inverse as
\begin{equation}
    h_{a}^{\bar b}=\bar h_{\bar B\,a}^{\bar b}G^{\bar B D}\del_B{\cal G}\:,
\end{equation}
where $\cal G=\log(\kappa)$, and the derivative $\del_B$ is with respect to $Z^A$.

\section{Connections and Curvature}
\label{app:connection}
We can use the complex structure deformation $h\in\Omega^{0,1}(T^{1,0}X)$, where $\vh\neq0$, to define a "Chern-type" connection $\tilde\nabla_a$ on holomorphic tangent and co-tangent indices. This connection was first defined in \cite{MurgasIbarra:2024phy}, and acts for example on $\alpha\in\Omega^{0,p}(T^{*1,0}X)$ as
\begin{equation}
    \tilde\nabla_a\alpha_b={h_b}^{\bar c}\partial_a\left({h_{\bar c}}^d\alpha_d\right)\:.
\end{equation}
Here ${h_b}^{\bar c}$ is the inverse of ${h_{\bar c}}^d$, which exists whenever $\vh\neq0$. Here $\vh$ is given as
\begin{equation}
    \Omega(h,h,h)=\tfrac{1}{3!}h^ah^bh^c\,\Omega_{abc}=\vh\,\bar\Omega\:,
\end{equation}
which, as shown in Appendix \ref{app:spacial}, is a constant on $X$. On $\beta\in\Omega^{0,p}(T^{1,0}X)$ the connection acts as
\begin{equation}
\tilde\nabla_a\beta^b={h_{\bar c}}^b\del_a\left({h_{c}}^{\bar c}\beta^c\right)\:.    
\end{equation}
A similar connection $\tilde\nabla_{{\bar a}}$ can be defined on anti-holomorphic indices.\footnote{Here one can either use $h$ and the inverse of $h$ or the complex conjugate of $h$ to define the connection. However, the corresponding two connections only agree for "special directions" as discussed in section \ref{sec:EndTtheory}.}

Clearly the connection is "metric" with respect to $h$. Furthermore, the connection also preserves the holomorphic top-form, that is
\begin{proposition}
    The holomorphic top-form $\Omega$ is $\tilde\nabla$-paralell. That is,
    \begin{equation}
        \tilde\nabla_a\Omega_{bcd}=0\:.
    \end{equation}
\end{proposition}
\begin{proof}
The proof is a straight forward computation.
\begin{equation}
    \tilde\nabla_a(\Omega_{bcd})={h_{b}}^{\bar b}{h_{c}}^{\bar c}{h_{d}}^{\bar d}\partial_a\left({h_{\bar b}}^{b}{h_{\bar c}}^{c}{h_{\bar d}}^{d}\Omega_{bcd}\right)\propto\vh{h_{b}}^{\bar b}{h_{c}}^{\bar c}{h_{d}}^{\bar d}\partial_a(\bar\Omega_{\bar b \bar c \bar d})=0\:,
\end{equation}
where we have used that ${h_{\bar b}}^{b}{h_{\bar c}}^{c}{h_{\bar d}}^{d}\Omega_{bcd}\propto\vh\Omega_{\bar b \bar c \bar d}$, and $\vh$ is point-wise constant.    
\end{proof}

The curvature of $\tilde\nabla$ is
\begin{equation}
    \tilde R{}_{\bar b a}{}^{c}{}_{d\phantom{\bar  b}\!\!}=\partial_{\bar b}\left({h_d}^{\bar c}\partial_a{h_{\bar c}}^c\right)\:,
\end{equation}
in close analogy with the curvature of the ordinary Chern connection. Let us make some further observations about this connection and its curvature. We first note that the curvature is "Ricci-flat". Indeed, we have
\begin{proposition}
\label{prop:Ricci}
The curvature of $\tilde\nabla$ satisfies
   \begin{equation}
    \tilde R{}_{\bar b a}{}^a{}_{d\phantom{\bar  b}\!\!}=0\:.
\end{equation} 
\end{proposition}
\begin{proof}
    This follows since
    \begin{equation}
    \tilde R{}_{\bar b a}{}^a{}_{d\phantom{\bar  b}\!\!}=\partial_{\bar b}\left({h_d}^{\bar c}\partial_a{h_{\bar c}}^a\right)=0\:.
\end{equation}
    The last equality follows since $h$ is divergence-free $\partial_a{h_{\bar c}}^a=0$, which is equivalent to $\del\Omega(h)=0$. 
\end{proof}
Furthermore, we also have
\begin{proposition}
\label{prop:TorsionFree}
   The connection symbols 
\begin{equation}
    \tilde\Omega{}_a{}^c{}_d={h_d}^{\bar c}\partial_a{h_{\bar c}}^c
\end{equation}
are torsion-free. That is
\begin{equation}
    \tilde\Omega{}_{[a}{}^c{}_{d]}=0\:.
\end{equation} 
Furthermore, the connection symbols are trace-free. That is
\begin{equation}
    \tilde\Omega_a{}^b{}_b=0\:.
\end{equation}
\end{proposition}
\begin{proof}
    If ${h_{\bar c}}^a$ is harmonic, then so is its inverse, given by the expression \eqref{eq:h-Inv}. It follows that
\begin{equation}
    \tilde\Omega{}_{[a}{}^c{}_{d]}={h_{[d}}^{\bar c}\partial_{a]}{h_{\bar c}}^c=-\del_{[a}{h_{d]}}^{\bar c}{h_{\bar c}}^c=0\:,
\end{equation}
where we use that the inverse $h^{\bar c}\in\Omega^{1,0}(T^{0,1}X)$ is $\del$-closed. It then also follows that
\begin{equation}
    \tilde\Omega_a{}^b{}_b={h_b}^{\bar c}\partial_a{h_{\bar c}}^b={h_a}^{\bar c}\partial_b{h_{\bar c}}^b=0\:,
\end{equation}
as $\partial_b{h_{\bar c}}^b=0$ when ${h_{\bar c}}^b$ is harmonic. 
\end{proof}

Finally, we also have
\begin{proposition}
    The curvature is also primitive with respect to $h$. That is, the curvature satisfies equation \eqref{eq:Instanton}, which is equivalent to
    \begin{equation}
        h^{a}\tilde R_a{}^b{}_c=0\:,
    \end{equation}
    where $\tilde R_a{}^b{}_c=\dd z^{\bar d}\tilde R_{a\bar d}{\,}^b{}_c$.
\end{proposition}
\begin{proof}
     Using that the connection symbols are torsion-free, we see that
\begin{equation}
    h^{a}\tilde R_a{}^b{}_c=-\tilde R_a{}^b{}_ch^{a}=-\tilde R_c{}^b{}_ah^{a}=-(\delb\circ\tilde\nabla_c-\tilde\nabla_c\circ\delb)h^{b}=0\:,
\end{equation}
where we use that $h$ is $\tilde\nabla$-parallel, and $\delb$-closed. 
\end{proof}

\bibliographystyle{JHEP.bst}
\bibliography{refs}

\end{document}